%% file: main.tex
\documentclass[12pt]{article}
\usepackage[left=2.7cm,top=3.5cm,right=2cm,bottom=2.5cm]{geometry}
\newcommand{\parder}[2]{\frac{\partial #1}{\partial #2}}
\newcommand{\parsecder}[3]{\frac{\partial^2 #1}{\partial #2\; \partial #3}}
\newcommand{\shp}[2]{N^{#1}_{#2}(\bx)}
\newcommand{\bzero}{\mathbf{0}}
\newcommand{\del}{\nabla}

\newcommand{\bff}{\boldsymbol{f}}
\newcommand{\bg}{\boldsymbol{g}}

\newcommand{\bk}{\boldsymbol{k}}
\newcommand{\bl}{\boldsymbol{l}}

\newcommand{\bq}{\boldsymbol{q}}
\newcommand{\br}{\boldsymbol{r}}
\newcommand{\bs}{\boldsymbol{s}}

\newcommand{\bu}{\boldsymbol{u}}
\newcommand{\bv}{\boldsymbol{v}}
\newcommand{\bx}{\boldsymbol{\textbf{x}}}
\newcommand{\bA}{\boldsymbol{A}}
\newcommand{\bB}{\boldsymbol{B}}
\newcommand{\bC}{\boldsymbol{C}}

\newcommand{\bH}{\boldsymbol{H}}

\newcommand{\bK}{\boldsymbol{K}}
\newcommand{\bKbar}{\bar{\bK}}
\newcommand{\bxbar}{\bar{\boldsymbol{x}}}
\newcommand{\bffbar}{\bar{\bff}}

\newcommand{\bR}{\boldsymbol{R}}

\newcommand{\bphi}{\boldsymbol{\phi}}
\newcommand{\bLam}{\boldsymbol{\Lambda}}

\newcommand{\ubarn}{\bar{u}}
\newcommand{\uhat}{\hat{\bu}}
\newcommand{\fhat}{\hat{\bff}}
\newcommand{\ghat}{\hat{\bg}}
\newcommand{\phihat}{\hat{\bphi}}

\newcommand{\phibar}{\bar{\phi}}
\newcommand{\mubar}{\bar{\mu}}

\newcommand{\omegbar}{\bar{\omega}}
\newcommand{\alpj}{\omegbar^{\alpha_{J}}}
\newcommand{\betj}{\omegbar^{\beta_{J}}}

\newcommand{\dalpjn}{\delta\omega^{\alpha_{J}}}
\newcommand{\dbetjn}{\delta\omega^{\beta_{J}}}

\newcommand{\alpjn}{\omega^{\alpha_{J}}}
\newcommand{\betjn}{\omega^{\beta_{J}}}

\newcommand{\ualp}{\ubarn^{2\alpha}}
\newcommand{\ubet}{\ubarn^{2\beta}}
\newcommand{\ualpha}{u^{2\alpha}}
\newcommand{\ubeta}{u^{2\beta}}
\newcommand{\ualphabeta}{u^{2(\alpha+\beta)}}
\newcommand{\ualpbet}{\ubarn^{2(\alpha+\beta)}}

\newcommand{\intomega}{\int_{\varOmega}}

\newcommand{\intR}{ \int_{\mathbb{R}^{3}}}

\newcommand{\dx}{\,d\bx}
\newcommand{\Kalpbet}{K(|\bx - \bx'|)}

\usepackage{amsmath}
\usepackage{graphicx}
\usepackage{multirow}
\usepackage{fancyhdr, ifpdf}
\usepackage{amsxtra}
\usepackage{stmaryrd}
\usepackage{mathrsfs}
\usepackage{psfrag}
\usepackage{subfigure}
\usepackage{rotating}
\usepackage{setspace}
\usepackage{float}
\usepackage{caption}
\usepackage{amsfonts}
\usepackage{amsbsy}
\usepackage{amssymb}
\usepackage{amscd}
\usepackage{amsthm}
\usepackage{color}
\newtheorem{prop}{Proposition}[section]
\newtheorem*{remark}{Remark}
\definecolor{hellgruen}{rgb}{0.2,0.7,0.2}

\newcommand{\cn}{\color{black}}

\date{}
\begin{document}
\author{\it{Phani Motamarri\,\,$^a$,  Mrinal Iyer\,\,$^a$,  Jaroslaw Knap\,\,$^b$, Vikram Gavini\,\,$^{a}\footnote{Corresponding author}$}\\
 \normalsize $^a$ Department of Mechanical Engineering, University of Michigan, Ann Arbor, MI 48109, USA \\
 \normalsize $^b$ U.S. Army Research Labs, Aberdeen Proving Grounds, Aberdeen, MD 21001, USA}
\title{\LARGE{\textbf{Higher-order adaptive finite-element methods for orbital-free density functional theory}}}
\maketitle
\abstract{\input{abstract.tex}}
\section{Introduction}
\input{intro.tex}

\section{Formulation}\label{formulation}
\input{formulation.tex}

\section{\emph{A priori} Mesh Adaption}\label{meshadapt}
\input{error.tex}

\section{Numerical Implementation}\label{NumericalImplementation}
\input{Numericalimplementation.tex}

\section{Numerical Results}\label{NumericalResults}
\input{NumericalResults.tex}

\section{Conclusions}\label{concl}
\input{conclusions.tex}

\section*{Acknowledgements}
\input{acknowledgements.tex}

\appendix
\input{appendix.tex}

\end{document}

%% file: abstract.tex
In the present work, we investigate the computational efficiency afforded by higher-order finite-element discretization of the saddle-point formulation of orbital-free density-functional theory. We first investigate the robustness of viable solution schemes by analyzing the solvability conditions of the discrete problem. We find that a staggered solution procedure where the potential fields are computed consistently for every trial electron-density is a robust solution procedure for higher-order finite-element discretizations. We next study the numerical convergence rates for various orders of finite-element approximations on benchmark problems. We obtain close to optimal convergence rates in our studies, although orbital-free density-functional theory is nonlinear in nature and some benchmark problems have Coulomb singular potential fields. We finally investigate the computational efficiency of various higher-order finite-element discretizations by measuring the CPU time for the solution of discrete equations on benchmark problems that include large Aluminum clusters. In these studies, we use mesh coarse-graining rates that are derived from error estimates and an {\it a priori} knowledge of the asymptotic solution of the far-field electronic fields. Our studies reveal a significant 100-1000 fold computational savings afforded by the use of higher-order finite-element discretization, alongside providing the desired chemical accuracy. We consider this study as a step towards developing a robust and computationally efficient discretization of electronic structure calculations using the finite-element basis.

%% file: intro.tex
Electronic structure calculations have been successful in accurately predicting a wide range of material properties over the course of past few decades. These predictions range from the accurate description of phase transformations in various materials to the electronic and optical properties of nanostructures. The predictive capability of electronic structure calculations can be attributed to the fact that they are derived from many-body quantum-mechanics and incorporate much of the fundamental physics with little empiricism. One of the most popular and widely used electronic structure theories is the Kohn-Sham approach to density functional theory (DFT)~\cite{kohn}. It is based on the Hohenberg-Kohn theorem~\cite{Hohenberg} which states that there is a one-to-one correspondence between the ground-state electronic wave-function of a quantum-mechanical system with $N$ interacting electrons and the ground-state electron-density. Thus, based on the Hohenberg-Kohn result, the ground-state properties of any quantum-mechanical system can be described by an energy functional of the electron-density. While the existence of such an energy functional to describe the ground-state properties is known, its functional representation is unknown to date. The Kohn-Sham approach circumvents this challenge by considering an equivalent system of non-interacting electrons moving in an effective mean field governed by the electron-density. While, in principle, this non-interacting independent particle description is exact for ground-state properties, it is formulated in terms of an unknown exchange-correlation functional of the electron-density for which various approximate models are used. Further, the Kohn-Sham approach is computationally expensive for large scale simulations (thousands of electrons or more), since it requires an expensive evaluation of the kinetic energy functional that involves the computation of the $N$ independent single-electron wave-functions self-consistently, which scales as $O(N^3)$. In order to reduce the computational complexity of electronic structure calculations using DFT, approximate models for the kinetic energy density functionals have been proposed and used in electronic structure calculations (cf.~e.g.,~\cite{parr,Wang1992,Enrico,wang1,wang2}). This approach is referred to as orbital-free DFT where the ground-state energy is explicitly described as a functional of the electron-density. Various numerical studies have shown that the orbital-free approximation to DFT is reasonable for systems with electronic structure close to a free-electron gas (cf.~e.g.~\cite{wang1,wang2,Huang}), for e.~g. simple metals and Aluminum, but provides an inaccurate description for covalently bonded and ionic systems. Irrespective of whether one is using the Kohn-Sham or orbital-free approach, realistic simulations of large-scale material systems with DFT are still very demanding, and numerical algorithms which are robust, computationally efficient and scalable on parallel computing architectures are always desirable to enable simulations at larger scales and on more complex systems.

Plane-wave basis has been the most popular basis set used in DFT calculations to date as it renders an efficient computation of electrostatic interactions naturally through Fourier transforms~\cite{VASP,CASTEP,ABINIT,PROFESS}. However, a plane-wave basis often restricts the geometry of simulation domains to periodic systems, which is incompatible with the displacement fields produced by most crystalline defects. Further, a plane-wave basis provides a uniform spatial resolution which is computationally inefficient in the study of defects and non-periodic systems like molecules and clusters---a higher resolution is often desired to describe specific regions of interest, such as a defect-core, but a coarser resolution suffices elsewhere. Moreover, the non-locality of a plane-wave basis affects the scalability of computations on massive parallel computing architectures. Thus, there is an increasing thrust towards using real-space techniques for electronic structure calculations (cf.~\cite{Beck} and references therein for a comprehensive overview). Among the real-space techniques, many efforts have focused on developing the finite-element basis for electronic structure calculations as it can accommodate unstructured coarse-graining of the basis set, can handle complex geometries and boundary conditions, and is easy to parallelize. We refer to ~\cite{white,tsuchida1995,tsuchida1996,tsuchida1998,pask1999,pask2001,carlos2006,ofdft,qcofdft,Zhou2008,suryanarayana2010non,Lin,bylaska,lehtovaara} and references therein for a comprehensive overview. The convergence of the finite-element approximation using a variational notion of $\Gamma-$convergence has been shown for the case of orbital-free DFT (with Thomas-Fermi-von Weizsacker family of kinetic energy functionals)~\cite{ofdft} as well as for Kohn-Sham DFT~\cite{suryanarayana2010non}. Theoretical rates of convergence for the finite-element approximation of orbital-free DFT problem with pseudopotentials have been recently estimated in~\cite{ortner}. We also refer to~\cite{zhou,Cances} for a numerical analysis of a class of nonlinear eigenvalue problems, which includes the orbital-free DFT problem.

While the finite-element basis is more versatile in comparison to a plane-wave basis~\cite{pask1999}, prior studies have shown that the number of basis functions in a linear finite-element discretization required to achieve chemical accuracy in electronic structure calculations is very large~(cf. e.~g. \cite{bylaska, Hermannson}), thus increasing the computational cost of calculations. The use of higher-order finite-element discretizations was suggested as a possible solution to alleviate this degree of freedom disadvantage of linear finite-elements in electronic structure calculations---cf.~\cite{batcho} in the context of spectral elements and \cite{lehtovaara} in the context of hierarchical higher-order finite-element discretizations. However, using higher-order finite-element discretizations also increase the per basis function computational cost due to the use of higher-order numerical quadrature rules and an increase in the bandwidth of the matrix which affects the efficiency of the iterative solution schemes. Further, the rates of convergence for higher-order finite-element approximations of electronic structure theories, while numerically studied~\cite{pask1999,pask2005} and mathematically analyzed~\cite{ortner,zhou,Cances} in the context of pseudopotential calculations, are not completely understood in the case of all-electron calculations. Thus, in the present work, we conduct a systematic study to investigate the viability and the computational efficiency afforded by higher-order finite-element approximations in electronic structure calculations using orbital-free DFT as a model theory. We consider this investigation as an important step towards our final goal of studying the effectiveness of higher-order finite-element approximations in Kohn-Sham DFT.

A robust solution procedure and suitably graded finite-element meshes are key to examining the computational efficiency of higher-order finite-element discretizations in orbital-free DFT. To this end, we first investigate the robustness of viable solution schemes by analyzing the solvability conditions of the discrete orbital-free DFT problem and subsequently propose an {\it a priori} mesh adaption scheme to construct suitably graded finite-element meshes that will be used in the numerical investigations. We begin with the local real-space formulation as presented in~\cite{ofdft,bala}, where the problem of computing the ground-state energy is reformulated as a local saddle-point problem in the electronic fields comprising of the electron-density, electrostatic potential, and the kernel potentials describing the extended interactions in kinetic energy functionals. The discrete non-linear equations resulting from the finite-element discretization of the saddle-point problem can be solved using two different approaches: (i) a simultaneous solution of all fields using quasi-Newton methods; (ii) a staggered approach where the electron-density is solved using quasi-Newton methods, whereas the potential fields are computed consistently for every trial electron-density using iterative linear solvers. We investigate the robustness of these approaches by analyzing the solvability conditions for the linearized incremental problem corresponding to a single step within the iterative scheme of the nonlinear problem. We consider the scheme to be well-behaved if each iterative update is uniquely solvable. While the staggered approach guarantees solvability within the vicinity of the exact solution and is robust for any order of discretization of electronic fields (argument based on Proposition 3.3 of \cite{ortner}), the same is not true with a simultaneous solution procedure. The discrete equations in simultaneous scheme are not solvable if the dimension of the vector-space approximating electron-density far exceeds that of electrostatic potential and the Hessian matrix corresponding to electron-density is rank deficient to a certain degree which depends on the initial guess. Having examined the best solution procedure to solve the discrete equations (staggered approach), we next turn towards determining an optimal finite-element mesh for the orbital-free DFT problem. To this end, we derive error estimates in the energy as a function of the characteristic mesh-size distribution $h(\bx)$ and the electronic fields comprising of the electron-density, electrostatic potential and kernel potentials. We remark that error estimates have been rigorously derived in recent studies for the orbital-free DFT problem~\cite{ortner,zhou,Cances}, though the form of these estimates is not useful for developing mesh adaption schemes as these studies primarily focused on proving the convergence of the finite-element approximation and determining the convergence rates. We use the derived error estimates in conjunction with the mesh adaption scheme proposed in~\cite{radio} to determine the optimal coarse-graining rates for the finite-element discretizations in terms of the degree of the interpolating polynomial and the exact solution fields in the orbital-free DFT problem. As the exact solution fields are {\it a priori} unknown, we use the asymptotic solutions of the electronic fields away from the nuclei, which are {\it a priori} known, to determine the coarse-graining rates for the finite-element discretizations that will subsequently be used in the numerical studies. We note that the finite-element meshes thus determined are optimal away from nuclei, but not necessarily optimal in the vicinity of the nuclei. We remark that obtaining optimal meshes away from the nuclei is still very useful for efficiently resolving the vacuum in non-periodic calculations.

We next turn towards the numerical investigation of the saddle-point orbital-free DFT problem in the finite-element basis. The various finite-elements used in our study include tetrahedral elements up to second-order, hexahedral elements up to fourth-order as well as spectral finite-elements~\cite{patera1984spectral} of order three and four. We begin our study by computing the numerical rates of convergence for all the elements. These convergence studies are conducted on three benchmark problems comprising of: (i) Hydrogen atom, which represents a linear problem with a singular potential; (ii) Helium atom, which represents a non-linear problem with a singular potential; (iii) Aluminum atom, which represents a non-linear problem with a pseudopotential. In these studies, as well as those to follow, we use the meshes obtained from the proposed {\it a priori} mesh adaption scheme. These benchmark studies show rates of convergence in energy of $ O(h^{2k})$ for a finite-element whose degree of interpolation is $k$, which denote optimal rates of convergence determined by the error analysis of the orbital-free DFT problem~\cite{ortner,zhou,Cances}. We remark that these numerical studies show optimal rates of convergence even for singular Coulomb potentials (Hydrogen and Helium atom) which were not analyzed in~\cite{ortner,zhou,Cances}. To the best our knowledge, an error analysis of orbital-free DFT considering singular Coulomb potentials is an open question to date. We next conduct a systematic study to assess the computational efficiency afforded by higher-order finite-element discretizations. To this end, we use the same benchmark problems and measure the CPU time for the solution of the orbital-free DFT problem to various relative accuracies for all the finite-elements considered in this study. We find that higher-order elements are computationally more efficient, and can provide significant computational savings, especially in the regime of chemical accuracy. For instance, a staggering 100-1000 fold computational savings are realized with the use of higher-order elements resulting in a relative energy error of $10^{-4}$ in comparison to linear finite-elements with a much higher relative error of around $10^{-2}$. We also observe a 2-3 fold computational savings in using spectral finite-elements over conventional finite-elements of the same order. To further assess the performance of higher-order elements, we conducted simulations on a series of aluminum clusters ranging from 1 face-centered-cubic (fcc) unit cell to $5\times5\times5$ unit cells, and these studies also reflect a similar computational advantage afforded by the use of higher-order elements.

The remainder of the paper is organized as follows. Section~\ref{formulation} describes the saddle-point formulation of the orbital-free DFT problem and summarizes the key results from the analysis of solvability conditions of the linearized system of equations corresponding to the discrete orbital-free DFT problem, whose detailed analysis is provided in appendix~\ref{solve}. Section~\ref{meshadapt} presents the error estimates for the finite-element discretization of orbital-free DFT and uses these estimates to present an {\it a priori} mesh adaption scheme. Sections~\ref{NumericalImplementation} and \ref{NumericalResults} describe the numerical implementation of the method and present the numerical examples which demonstrate the computational efficiency afforded by higher-order finite-element discretizations in electronic structure calculations. We conclude with a short discussion in Section~\ref{concl}.

%% file: formulation.tex
In this section, we present the variational formulation underlying the orbital-free DFT problem and a discretization of the formulation using a finite-element basis. We also discuss the main results from the analysis of solvability conditions corresponding to the discrete orbital-free DFT problem, whose detailed analysis is provided in appendix~\ref{solve}. If $N$ and $M$  denote the number of electrons and the number of atoms in a charge neutral system respectively, and $u = \sqrt{\rho}$ denotes the square-root electron-density such that $\int u^2 \dx = N$, then the energy of the system described by density functional theory is given by (cf. e.~g.~\cite{parr,finnis})
\begin{equation}\label{engy}
E(u,\bR) = T_{s}(u) + E_{xc}(u) + E_{H}(u)+E_{ext}(u,\bR) + E_{zz}(\bR),
\end{equation}
where  $\bR = \{\bR_1,\cdots,\bR_M\}$ is the position vector denoting the nuclear coordinates in the system; $T_s$ is the kinetic energy of the non-interacting electrons; $E_{xc}$ is the exchange correlation energy; $E_H$ denotes the classical electrostatic interaction energy of the electron-density also referred to as Hartree energy; $E_{ext}$ is the interaction energy with the external field, $V_{ext}$, induced by nuclear charges and $E_{zz}$ is the repulsive energy between nuclei. In orbital-free DFT, $T_s$ is modeled using explicit functional forms of electron-density, and we employ, in our study, the Thomas-Fermi-von Weizscaker(TFW) family of functionals~\cite{parr}, as well as, the more accurate Thomas-Fermi-von Weizscaker functional with kernel energies~\cite{wang1,wang2}. In what follows, we present the variational formulations for each of these cases.

\subsection{The Variational Problem}
The TFW family of functionals have the following representation
\begin{equation}
T_s(u) = C_{F}\intR u^{10/3}(\bx) \dx + \frac{\lambda}{2}\intR |\del u(\bx)|^{2} \dx\,,
\end{equation}
where $C_{F} = \frac{3}{10}(3\pi^2)^{2/3}$ and $\lambda$ is a parameter. Though different values of $\lambda$ are found to work better in different cases, $\lambda = \frac{1}{5}$ is chosen in this study since this was found to be an optimal choice for various atomic systems~\cite{parr}.
$E_{xc}$, the exchange correlation energy in functional~\eqref{engy}, is modeled using the local density approximation(LDA)~\cite{alder,perdew} and is given by
\begin{equation}
E_{xc}(u) = \intR \varepsilon_{xc}(u(\bx))u^{2}(\bx)\dx,
\end{equation}
where $\varepsilon_{xc} = \varepsilon_{x} + \varepsilon_{c}$ is the exchange and correlation energy
per electron given by
\begin{equation}
\varepsilon_x(u) = -\frac{3}{4}\left(\frac{3}{\pi}\right)^{1/3}u^{2/3},
\end{equation}
\begin{equation}
\varepsilon_c(u) = \begin{cases}
&\frac{\gamma}{(1 + \beta_1\sqrt(r_s) + \beta_2r_s)}\;\;\;\;\;\;\;\;\;\;\;\;\;\;\;\;\;\;\;\;\;\;\;r_s\geq1\\
&A\,\log r_s + B + C\,r_s\log r_s + D\,r_s\;\;\;\;\;\;\;\;r_s\,<\,1,
\end{cases}
\end{equation}
where $r_s = (\frac{3}{4\pi u^2})^{1/3}$. The values of constants used in this study are those of an unpolarized medium, and are given by $\gamma_u$ = -0.1471, $\beta_{1u}$ = 1.1581, $\beta_{2u}$ = 0.3446, $A_{u}$ = 0.0311, $B_u$ = -0.048, $C_u$ = 0.0014, $D_u$ = -0.0108.
The electrostatic interactions in the functional~\eqref{engy} have the following form:
\begin{align}
E_{H}(u) &= \frac{1}{2}\intR\intR\frac{u^2(\bx)u^2(\bx')}{|\bx - \bx'|}\dx\dx',\label{hartree}\\
E_{ext}(u,\bR) &= \intR u^2(\bx) V_{ext}(\bx) \dx = \sum_{I=1}^{M}\intR u^2(\bx) \frac{Z_I}{|\bx-\bR_I|} \dx ,\label{external}\\
E_{zz} &= \frac{1}{2}\sum_{I=1}^{M}\sum_{\substack{J=1\\J \neq I}}^{M} \frac{Z_I Z_J}{|\bR_I-\bR_J|}. \label{repulsive}
\end{align}
The electrostatic interaction terms as expressed in equations~\eqref{hartree}-\eqref{repulsive} are nonlocal in real-space, and, for this reason, evaluation of electrostatic energy is the computationally expensive part of the calculation. Following the approach in~\cite{ofdft}, the electrostatic interaction energy can be reformulated as a local variational problem in electrostatic potential by observing that $\frac{1}{|\bx|}$ is the Green's function of the Laplace operator. To this end, we consider a nuclear charge $Z_I$ located at $\bR_I$ as a bounded regularized charge distribution $Z_I\tilde{\delta}_{\bR_I}(\bx)$ having a support in a small ball around $\bR_{I}$ and a total charge $Z_I$. The nuclear repulsion energy can subsequently be represented as
 \begin{equation}
E_{zz}(\bR) = \frac{1}{2}\intR\intR\frac{b(\bx)b(\bx')}{|\bx - \bx'|}\dx\dx',
\end{equation}
where $b(\bx) = \sum_{I = 1}^{M}Z_I\tilde{\delta}_{\bR_{I}}(\bx)$. We remark that, while this differs from the expression in equation~\eqref{repulsive} by the self-energy of the nuclei, the self-energy is an inconsequential constant depending only on the nuclear charge distribution, and is explicitly evaluated and subtracted from the total energy in numerical computations. Subsequently, the electrostatic interaction energy, up to a constant self-energy, is given by the following variational problem:
\begin{align}\label{elReformulation}
&\frac{1}{2}\intR\intR\frac{u^2(\bx)u^2(\bx')}{|\bx - \bx'|}\dx\dx' +
\intR u^2(\bx) V_{ext}(\bx) \dx + \frac{1}{2}\intR\intR\frac{b(\bx)b(\bx')}{|\bx - \bx'|}\dx\dx' \notag\\
&= -\inf_{\phi \in \mathcal{Y}} \left\{\frac{1}{8\pi}\intR |\del \phi(\bx)|^2 \dx - \intR (u^2(\bx) + b(\bx))\phi(\bx)\dx\right\},
\end{align}
where $\phi(\bx)$ denotes the trial function for the total electrostatic potential due to the electron-density and the nuclear charge distribution and $\mathcal{Y}$ is a suitable function space that guarantees existence of a minimizer.
Using equation ~\eqref{elReformulation}, the energy functional ~\eqref{engy} is reformulated in a local form given by
\begin{equation}
E(u,\bR) = \sup_{\phi \in \mathcal{Y}} L(u,\phi,\bR),
\end{equation}
where
\begin{equation}
\begin{split}
L(u,\phi,\bR) &= C_F\intR u^{10/3}(\bx)\dx + \frac{\lambda}{2}\intR |\del u(\bx)|^{2} \dx + \intR \varepsilon_{xc}(u^2(\bx))u^2(\bx)\dx \\
&- \frac{1}{8\pi}\intR |\del \phi(\bx)|^2 \dx + \intR (u^2(\bx) + b(\bx))\phi(\bx)\dx.
\end{split}
\end{equation}
Finally, the problem of determining the ground-state energy and electron-density for given positions of the nuclei can be expressed as the saddle-point problem
\begin{equation}\label{infsup}
\inf_{\raisebox{-0.7ex}{$\scriptstyle u \in \mathcal{X}$}} \sup_{\phi\in \mathcal{Y}} L(u,\phi,\bR)\qquad \mbox{subject to:\,\,} \intR u^2(\bx) \dx = N,
\end{equation}
where $\mathcal{X}$ denotes a suitable function space that guarantees the existence of minimizers. We remark that numerical computations are done on a finite-domain, $\varOmega$, which in non-periodic calculations corresponds to a large enough domain containing the compact support of $u$ for all practical purposes, and which in periodic calculations corresponds to the domain defining the unit cell. $\mathcal{X}=\mathcal{Y}=H^1_0(\Omega)$ for non-periodic problems and $\mathcal{X}=H^1_{per}(\Omega)$ and $\mathcal{Y} = \{\phi:\phi\in H^1_{per}(\Omega), \intomega{\phi}=0\}$  for periodic problems guarantees existence of solutions for the saddle-point in~\eqref{infsup}, and refer to~\cite{ofdft} for further details. We also refer to~\cite{lieb} for a comprehensive mathematical analysis of various flavors of orbital-free DFT models on unbounded domains.

We now turn towards the variational formulation of the orbital-free DFT model with kernel energies.  The kinetic energy of non-interacting electrons, $T_s(u)$, modeled using kernel energies has the following representation
\begin{equation}
T_s(u) = C_F\intR u^{10/3}(\bx) \dx + \frac{1}{2}\intR |\del u(\bx)|^2\dx + T_k(u),
\end{equation}
where $T_k$ is the non-local kernel energy and is given by
\begin{equation}
T_k(u) = \intR \intR u^{2\alpha}(\bx)K(|\bx-\bx'|,u(\bx),u(\bx'))u^{2\beta}(\bx')\dx\dx' \label{ker_engy}.
\end{equation}
The kernel $K$ is chosen such that $T_s$ satisfies the Lindhard susceptibility function for an appropriate choice of the parameters $\alpha$, $\beta$~\cite{wang1,wang2}. Further, these kernels are referred to as density dependent(DD) or density independent(DI) kernels based on the dependence or independence of kernel $K$ on $u$. It is a common practice to decompose the DD kernels into a series of DI kernels through a Taylor expansion about a reference density~\cite{wang2,kaxiras} and a similar approach is followed in this work.

Similar to the electrostatic interaction energy, the kernel energy has interactions that are extended in real space. The local reformulation of non-local kernel energies into a saddle point problem has been addressed in \cite{bala}, and the formulation is provided here for the sake of completeness. Defining the potentials $V_{\alpha}(\bx)$ and $V_{\beta}(\bx)$ as
\begin{subequations}\label{ker_pot}
\begin{gather}
V_{\alpha}(\bx) = \intR \Kalpbet \ualpha(\bx')\dx' \,,  \\
V_{\beta}({\bx}) = \intR \Kalpbet \ubeta(\bx') \dx' \,,
\end{gather}
\end{subequations}
we take the fourier transform to obtain $\hat{V}_{\alpha}(\bk) = \hat{K}(\bk)\hat{u^{2\alpha}}(\bk)$ and $\hat{V}_{\beta}(\bk) = \hat{K}(\bk)\hat{u^{2\beta}}(\bk)$. As shown in~\cite{kaxiras}, $\hat{K}$ can approximated to a very good accuracy using a sum of partial fractions of the following form:
\begin{equation}
\hat{K}(\bk) \approx \sum_{J=1}^{m} \frac{A_J|\bk|^2}{|\bk|^2 + B_J} \,,
\end{equation}
where $A_{J}, B_{J}, J = 1\cdots m$ are complex constants determined using a best fit approximation. Using this approximation, the potentials defined in~\eqref{ker_pot} take the form
\begin{subequations}
\begin{gather}
V_{\alpha}(\bx) = \sum_{J=1}^{m} \left[\alpjn(\bx) + A_J\ualpha(\bx)\right]\,, \\
V_{\beta}(\bx) = \sum_{J=1}^{m}\left[\betjn(\bx) + A_J\ubeta(\bx)\right]\,,
\end{gather}
\end{subequations}
where $\alpjn(\bx)$ and $\betjn(\bx)$ are referred to as kernel potentials and are solutions to the following Helmholtz equations for $J=1\ldots m$:
\begin{subequations}\label{helm_de}
\begin{gather}
-\del^2\alpjn + B_J\alpjn + A_JB_J\ualpha = 0 \,,  \\
-\del^2\betjn + B_J\betjn + A_JB_J\ubeta = 0\,.
\end{gather}
\end{subequations}
In the notation introduced in the above equation, we note that $\alpha_J$ in $\alpjn$ is only a superscript, and should not be interpreted as the power of $\omega$ (and likewise for $\betjn$). For convenience of notation, we define $\mathbf{\tilde{\omega}}^{\alpha}=\{\omega^{\alpha_1},\omega^{\alpha_2},\ldots,\omega^{\alpha_m}\}$ and $\mathbf{\tilde{\omega}}^{\beta}=\{\omega^{\beta_1},\omega^{\beta_2},\ldots,\omega^{\beta_m}\}$, which denote the vectors containing the corresponding kernel potentials. The Helmholtz equations in equation~\eqref{helm_de} can be expressed in a variational form which allows us to reformulate the non-local energies in~\eqref{ker_engy} as the following local saddle-point problem:
\begin{equation}\label{ker}
T_k(u) = \inf_{{\raisebox{-0.7ex}{$\scriptstyle\mathbf{\tilde{\omega}}^{\alpha}\in \mathcal{Z}$}}}\,\sup_{\mathrm{\tilde{\omega}}^{\beta} \in \mathcal{Z}}\bar{L}(u,\mathbf{\tilde{\omega}}^{\alpha},\mathbf{\tilde{\omega}}^{\beta})\,,
\end{equation}
where
\begin{equation}
\begin{split}
\bar{L}(u,\mathbf{\tilde{\omega}^{\alpha}},\mathbf{\tilde{\omega}^{\beta}}) &= \sum_{J=1}^{m}\Bigl\{\intR \Bigl[\frac{1}{A_J\;B_J}\del\alpjn\cdot\del\betjn + \frac{1}{A_J}\alpjn\betjn \\
&+ \betjn\ualpha + \alpjn\ubeta + A_J\ualphabeta\Bigr]\dx\Bigr\}.
\end{split}
\end{equation}
In the saddle-point problem in equation~\eqref{ker}, the function space is chosen to be $\mathcal{Z}=(H^1_{per}(\Omega))^m$ for periodic problems where $\varOmega$ denotes the unit-cell in the periodic calculation, and we do not treat non-periodic problems with kernel functionals. Using the local reformulation of the electrostatic interactions and kernel energies from equations~\eqref{elReformulation} and~\eqref{ker}, the problem of computing the ground-state properties in orbital-free DFT, for a fixed position of atoms, can now be expressed as the following local saddle-point problem in real-space on bounded domains as:
\begin{equation}\label{infsup_kernel}
\inf_{\raisebox{-0.7ex}{$\scriptstyle u \in \mathcal{X}$}}\,\sup_{\phi\in\mathcal{Y}}\,\inf_{{\raisebox{-0.7ex}{$\scriptstyle\mathbf{\tilde{\omega}}^{\alpha}\in \mathcal{Z}$}}}\,\sup_{\mathbf{\tilde{\omega}}^{\beta}\in\mathcal{Z}}\, \tilde{L}(u,\phi,\mathbf{\tilde{\omega}}^{\alpha},\mathbf{\tilde{\omega}}^{\beta},\bR)\qquad\mbox{subject to:\,\,}\intomega u^2 \dx = N,
\end{equation}
where
\begin{equation}
\tilde{L}(u,\phi,\mathbf{\tilde{\omega}}^{\alpha},\mathbf{\tilde{\omega}}^{\beta},\bR) = L(u,\phi,\bR) + \bar{L}(u,\mathbf{\tilde{\omega}}^{\alpha},\mathbf{\tilde{\omega}}^{\beta}).
\end{equation}
The local variational formulation of the ground-state electronic-structure described above forms the basis of the finite-element approximation schemes discussed subsequently.

\subsection{The Discrete Problem:}
 Let $X_{h}^{u}$ with dimension $n_{u}$, $X_{h}^{\phi}$ with dimension $n_{\phi}$, $X_{h}^{\omega}$ with dimension $n_{\omega}$ denote the finite-dimensional subspaces of $\mathcal{X}$ , $\mathcal{Y}$  and $\mathcal{Z}$   respectively. The discrete problem corresponding to~\eqref{infsup_kernel} is given by the following constrained saddle-point problem:
 \begin{equation}
 \begin{split}
 &\inf_{\raisebox{-0.7ex}{$\scriptstyle u_{h}\in X_{h}^{u}$}}\,\sup_{\phi_h \in X_{h}^{\phi}}\,\inf_{\raisebox{-0.7ex}{$\scriptstyle{\mathbf{\tilde{\omega}}^{\alpha}}_h\in X_{h}^{\omega}$}}\,\sup_{{\mathbf{\tilde{\omega}}^{\beta}}_h \in X_h^{\omega}}\, \tilde{L}(u_h,\phi_h,{\mathbf{\tilde{\omega}}^{\alpha}}_h,{\mathbf{\tilde{\omega}}^{\beta}}_h,\bR)\\
&  \;\;\;\;\text{subject to:\,\,}\intomega u_h^2 \dx = N.
  \end{split}\label{fem_kernel}
  \end{equation}
 Similarly, the discrete problem corresponding to~\eqref{infsup} is given by
 \begin{equation}
 \inf_{\raisebox{-0.7ex}{$\scriptstyle u_h \in X_h^{u}$}}\sup_{\phi_h \in X_{h}^{\phi}} L(u_h,\phi_h,\bR) \;\;\;\;\text{subject to}\intomega u_h^2 \dx = N \label{fem}.
 \end{equation}

The convergence of finite-element approximation for the orbital-free DFT model with TFW family of kinetic energy functionals was shown in~\cite{ofdft} using the notion of $\Gamma-$ convergence. Recent numerical analysis of the finite-element discretization also estimated the rates of convergence of the approximation, and we refer to ~\cite{ortner,zhou} for further details.

The higher-order finite element discretization of the saddle-point orbital-free DFT problem leads to a natural question: what is an efficient approach to solve the resulting nonlinear mixed-field finite-element equations in ~\eqref{fem_kernel} and ~\eqref{fem}? Two possible approaches to solve the problem are considered: (i) a simultaneous solution of all the fields using quasi-Newton methods; (ii) a staggered approach where the electron-density is solved using quasi-Newton methods, whereas the potential fields are computed consistently for every trial electron-density using iterative linear solvers. Appendix~\ref{solve} derives the solvability conditions of the linearized discrete finite-element equations arising out of ~\eqref{fem}, which provides insights into the robustness of these two approaches. The analysis of these solvability conditions indicates that, while the discrete system of finite-element equations in the staggered approach is solvable for any choice of interpolation spaces of the electronic fields, the simultaneous scheme imposes restrictions on the choice of the interpolation spaces of the fields involved. The necessary condition for the discrete system of finite-element equations to be solvable in a simultaneous scheme is that the rank of the Hessian matrix corresponding to root-electron-density ($\bA$) must be greater than or equal to $n_u - (n_{\phi} + 1)$ which depends on the initial guess of root-electron-density and electrostatic potential. Furthermore, if the finite element discretization is chosen such that $rank(\bA) \geq n_u - (n_{\phi} + 1)$ is always satisfied, the sufficiency condition for invertibility of the system matrix is given by (i) in Proposition~\ref{prop1}, which is difficult to check for any general discretization and guess of electronic fields.  Similar arguments can be extended to analyze the solvability conditions of the discrete finite-element problem  involving kernel potentials in ~\eqref{fem_kernel}. This analysis shows that the staggered solution procedure is a more robust approach to solve the mixed-field discrete finite-element equations corresponding to the orbital-free DFT problem and will be employed in the subsequent numerical investigations.

Next we derive the optimal coarse-graining rates for the finite-element meshes using the asymptotic nature of the solution fields in the orbital-free DFT problem.

%% file: error.tex
Following~\cite{radio}, we present an {\it a priori} mesh adaption scheme by minimizing the error in the finite-element approximation of the orbital-free DFT problem for a fixed number of elements in the mesh. To this end, we first seek a bound on the energy error $|E-E_h|$ as a function of the characteristic mesh-size, $h$, and the distribution of electronic fields. We remark that error estimates have been rigorously derived in recent studies for the orbital-free DFT problem~\cite{ortner,zhou,Cances}. However, the form of these estimates is not useful for developing mesh adaption schemes as these studies primarily focused on proving the convergence of the finite-element approximation and determining the convergence rates. In what follows, we present the derivation of error bound in terms of the root-electron-density and the electrostatic potential for the orbital-free DFT problem with TFW kinetic energy functionals. Similar error estimates for orbital-free DFT models with kernel energies are discussed in appendix~\ref{error_kernel}.
\subsection{Estimate of Energy Error}
Let ($\ubarn_h$, $\phibar_h$, $\mubar_h$) and ($\ubarn$, $\phibar$, $\mubar$) be the solutions of the discrete finite-element problem \eqref{fem} and the continuous problem \eqref{infsup} respectively for a given set of nuclear positions, where the nuclear charges are represented by a bounded regularized charge distribution $b(\bx)$.  The ground state energy in the discrete and the continuous formulations can be expressed as
\begin{equation*}
E_h(\ubarn_h,\phibar_h) \\
= \frac{\lambda}{2}\intomega |\del \ubarn_h|^{2} \dx \;+\; \intomega F(\ubarn_h) \dx\; - \; \frac{1}{8\pi}\intomega |\del \phibar_h|^2 \dx + \intomega (\ubarn_h^2 \; +\; b)\phibar_h \dx\,,
\end{equation*}
\begin{equation*}
E(\ubarn,\phibar) = \frac{\lambda}{2}\intomega |\del \ubarn|^{2} \dx \;+\; \intomega F(\ubarn) \dx\; - \; \frac{1}{8\pi}\intomega |\del \phibar|^2 \dx + \intomega (\ubarn^2 \; +\; b)\phibar \dx\,, \label{contengy}
\end{equation*}
where
\begin{equation*}
F(u) = C_F\;u^{10/3} + \varepsilon_{xc}(u^2)u^2\,.
\end{equation*}
\begin{prop}\label{prop3}
In the neighborhood of ($\ubarn$, $\phibar$, $\mubar$), the finite-element approximation error in the ground state energy can be bounded as:
\end{prop}
\begin{equation}
\begin{split}
|E_h - E| \leq &\frac{\lambda}{2} \left|\intomega |\del \delta u|^2 \dx\right| + \mubar\left|\intomega (\delta u)^2 \dx\right| + \frac{1}{2} \left|\intomega F''(\ubarn) (\delta u)^2 \dx \right|\\
&+ \frac{1}{8\pi} \left|\intomega |\del \delta \phi|^2 \dx \right|+ \left|\intomega (\delta u)^2\phibar \dx\right|  + 2 \left|\intomega \ubarn\; \delta u \;\delta \phi \dx\right|\,.
\end{split}
\end{equation}
\begin{proof}
We begin by expanding $E_h(\ubarn_h,\phibar_h)$ about the solution of the continuous problem, i.e $\ubarn_h = \ubarn + \delta u$ and $\phibar_h = \phibar + \delta \phi$. Using Taylor series expansion, we get
\begin{equation}
\begin{split}
E_h(\ubarn + \delta u, \phibar + \delta \phi) = &\frac{\lambda}{2}\intomega |\del(\ubarn + \delta u)|^{2} \dx + \intomega F(\ubarn + \delta u) \dx \;\\
 &- \;\frac{1}{8\pi}\intomega |\del (\phibar + \delta \phi)|^{2} \dx + \intomega \left((\ubarn + \delta u)^2 +b \right)\left(\phibar + \delta \phi\right)\dx\,,
 \end{split}
\end{equation}
which can be simplified to
\begin{equation}\label{taylor}
\begin{split}
&E_h(\ubarn_h,\phibar_h) = \frac{\lambda}{2} \intomega \left( |\del \ubarn|^2 + |\del \delta u|^2 + 2\del \ubarn\cdot  \del \delta u\right)\dx + \intomega F(\ubarn) \dx + \intomega F'(\ubarn) \delta u \dx \\&+ \frac{1}{2}\intomega F''(\ubarn) (\delta u)^2 \dx - \frac{1}{8\pi} \intomega \left( |\del \phibar|^2 + |\del \delta \phi|^2 + 2\del \phibar\cdot  \del \delta \phi\right)\dx + \intomega (\ubarn^2 + b)\phibar \dx \\&+ 2\intomega \ubarn\; \delta u\; \phibar \dx + \intomega (\ubarn^2 + b)\delta \phi \dx  + \intomega (\delta u)^2\phibar \dx +  2 \intomega \ubarn\; \delta u \;\delta \phi \dx + O(\delta u^{3}, \delta \phi^{3}, \delta u^{2} \delta \phi, \delta u \delta \phi^{2})\,.
\end{split}
\end{equation}
Since ($\ubarn$, $\phibar$, $\mubar$) satisfy the Euler-Lagrange equations of the functional \eqref{infsup}, we have
\begin{subequations}
\begin{gather}\label{euler}
\lambda \intomega \del \ubarn\cdot  \del \delta u \dx + \intomega F'(\ubarn) \delta u \dx + 2 \intomega \ubarn\; \delta u\; \phibar \dx = - 2\intomega \mubar\; \ubarn\; \delta u \dx\,,  \\
-\frac{1}{4\pi} \intomega \del \phibar \cdot \del \delta \phi + \intomega (\ubarn^2 + b)\delta \phi  = 0\,.
\end{gather}
\end{subequations}
Using equation~\eqref{taylor} and the above Euler-Lagrange equations, we get
\begin{equation}\label{error1}
\begin{split}
E_h - E &= \frac{\lambda}{2} \intomega |\del \delta u|^2 \dx - 2\mubar \intomega  \ubarn\; \delta u \dx + \frac{1}{2} \intomega F''(\ubarn) (\delta u)^2 \dx \\
&- \frac{1}{8\pi} \intomega |\del \delta \phi|^2 \dx + \intomega (\delta u)^2\phibar \dx + 2 \intomega \ubarn\; \delta u \;\delta \phi \dx + O(\delta u^{3}, \delta \phi^{3}, \delta u^{2} \delta \phi, \delta u \delta \phi^{2}) \,.
\end{split}
\end{equation}
Note that the constraint functional in the discrete form,
\begin{equation}
c(u_h) = \intomega u_h^2 \dx - N\,,
\end{equation}
is also expanded about the solution $\ubarn$, to give
\begin{equation}
c(u_h) = \intomega (\ubarn + \delta u)^2 \dx - N = \intomega (\ubarn^2 + (\delta u)^2 + 2\ubarn \delta u) \dx - N\,.
\end{equation}
Using
\begin{equation}
\intomega \ubarn^2 \dx = N\,,
\end{equation}
and $c(\ubarn_h) = 0$ we get
\begin{equation}\label{multiplier}
\intomega (\delta u)^ 2 = -2 \intomega \ubarn\; \delta u \dx\,.
\end{equation}
 Using equations~\eqref{error1} and \eqref{multiplier} we arrive at the following error bound in energy, upon neglecting third-order terms and beyond
\begin{equation}\label{error2}
\begin{split}
|E_h - E| \leq & \frac{\lambda}{2} \left|\intomega |\del \delta u|^2 \dx\right| + \mubar\left|\intomega (\delta u)^2 \dx\right| + \frac{1}{2} \left|\intomega F''(\ubarn) (\delta u)^2 \dx \right|\\
&+ \frac{1}{8\pi} \left|\intomega |\del \delta \phi|^2 \dx \right|+ \left|\intomega (\delta u)^2\phibar \dx \right| + 2 \left|\intomega \ubarn\; \delta u \;\delta \phi \dx\right|\,.
\end{split}
\end{equation}
\end{proof}
\begin{prop}\label{prop4}
The finite-element approximation error in proposition~\ref{prop3} expressed in terms of the approximation errors in root-electron-density and electrostatic potential is given by
\begin{equation}
|E_h - E| \leq C \left(\parallel\ubarn - \ubarn_h\parallel^{2}_{1,\Omega} +|\phibar - \phibar_h|^{2}_{1,\Omega} + \parallel\ubarn - \ubarn_h\parallel_{0,\Omega} \parallel\phibar - \phibar_h\parallel_{1,\Omega}\right)
\end{equation}
\end{prop}
\begin{proof}
We make use of  the following norms: $| \cdot |_{1,\Omega}$ represents the semi-norm in $H^{1}$ space, $\parallel\cdot \parallel_{1,\Omega}$  denotes the $H^{1}$ norm, $\parallel\cdot \parallel_{0,\Omega}$ and  $\parallel\cdot \parallel_{0,p,\Omega}$  denote the  standard $L^{2}$ and  $L^{p}$ norms respectively. All the constants to appear in the following estimates are positive and bounded. Firstly, we note that
\begin{equation}
\frac{\lambda}{2} \left|\intomega |\del \delta u|^2 \dx\right| = \frac{\lambda}{2} \intomega |\del \delta u|^2 \dx \leq C_1|\ubarn - \ubarn_h|^{2}_{1,\Omega}\,,
\end{equation}
\begin{equation}
\mubar \left|\intomega (\delta u)^2 \dx\right| = \mubar \intomega (\ubarn_h - \ubarn)^2 \dx \leq C_2 \parallel\ubarn - \ubarn_h\parallel ^{2}_{0,\Omega}\,.
\end{equation}
Using Cauchy-Schwartz and Sobolev inequalities, we arrive at the following estimate
\begin{align}
\frac{1}{2} \left|\intomega F''(\ubarn) (\delta u)^2 \dx \right| &\leq \frac{1}{2} \intomega  \left|F''(\ubarn) (\ubarn_h - \ubarn)^2\right| \dx \nonumber\\[0.1in]
&\leq C_3 \parallel F''(\ubarn) \parallel_{0,\Omega} \parallel(\ubarn - \ubarn_h)^2\parallel_{0,\Omega}\nonumber\\[0.1in]
&= C_3 \parallel F''(\ubarn) \parallel_{0,\Omega} \parallel\ubarn - \ubarn_h\parallel^{2}_{0,4,\Omega}\nonumber\\[0.1in]
&\leq \bar{C}_3 \parallel\ubarn - \ubarn_h\parallel^{2}_{1,\Omega}\,.
\end{align}
Further, we note
\begin{equation}
\frac{1}{8\pi} \left|\intomega |\del (\phibar_h - \phibar)|^2 \dx\right| \leq C_4 |\phibar - \phibar_h|^{2}_{1,\Omega}\,.
\end{equation}
Using Cauchy-Schwartz and Sobolev inequalities we arrive at
\begin{align}
\left|\intomega (\delta u)^2\phibar\dx\right| \leq \intomega \left|(\ubarn_h -\ubarn)^2\;\phibar\right| \dx &\leq \parallel \phibar \parallel_{0,\Omega} \parallel (\ubarn_h - \ubarn)^2\parallel_{0,\Omega}\nonumber\\
&\leq C_5 \parallel\ubarn - \ubarn_h\parallel^{2}_{0,4,\Omega}\nonumber\\
& \leq \bar{C}_5 \parallel\ubarn - \ubarn_h\parallel^{2}_{1,\Omega}\,.
\end{align}
Also note that
\begin{align}
\left|\intomega \ubarn\; \delta u \;\delta \phi \dx\right| &\leq \intomega \left|\ubarn (\ubarn_h - \ubarn)(\phibar_h - \phibar)\right|\dx \nonumber\\
&\leq \; \parallel \ubarn \parallel_{0,6,\Omega} \parallel\ubarn - \ubarn_h\parallel_{0,\Omega} \parallel\phibar - \phibar_h\parallel_{0,3,\Omega}\nonumber\\
&\leq C_6 \parallel\ubarn - \ubarn_h\parallel_{0,\Omega} \parallel\phibar - \phibar_h\parallel_{1,\Omega}\,,
\end{align}
where we made use of the generalized H\"older inequality in the first step and Sobolev inequality in the next. Using the bounds derived above, it follows that
\begin{equation}\label{error3}
|E_h - E| \leq C \left(\parallel\ubarn - \ubarn_h\parallel^{2}_{1,\Omega} + |\phibar - \phibar_h|^{2}_{1,\Omega} + \parallel\ubarn - \ubarn_h\parallel_{0,\Omega} \parallel\phibar - \phibar_h\parallel_{1,\Omega}\right)\,.
\end{equation}
\end{proof}
Now it remains to bound the finite-element discretization error with interpolation errors which can in turn be bounded with size of the finite-element mesh size $h$. This requires a careful analysis in the case of orbital-free DFT, which is a non-linear constrained problem, and has been discussed in~\cite{ortner}. Using the results from Theorem 3.2 in~\cite{ortner}, we bound the estimates in equation~\eqref{error3} using the following inequalities ( cf.~\cite{Ciarlet})
\begin{subequations}
\begin{gather}
\parallel\ubarn - \ubarn_h\parallel_{1,\Omega} \leq \bar{C}_0 \parallel\ubarn - u_I\parallel_{1,\Omega}\leq \tilde{C}_0\sum_{e}h_e^{k}|\ubarn|_{k+1,\Omega_e}\,,\\
\parallel\ubarn - \ubarn_h\parallel_{0,\Omega} \leq \bar{C}_1\parallel\ubarn - u_I\parallel_{0,\Omega} \leq \tilde{C}_1\sum_{e}h_e^{k+1}|\ubarn|_{k+1,\Omega_e}\,,\\
|\phibar - \phibar_h|_{1,\Omega} \leq \bar{C}_2 |\phibar - \phi_I|_{1,\Omega}\leq \tilde{C}_2\sum_{e}h_e^{k}|\phibar|_{k+1,\Omega_e}\,,
\end{gather}
\end{subequations}
where $k$ is the order of the polynomial interpolation, and $e$ denotes an element in the regular family of finite-elements~\cite{Ciarlet} with mesh-size $h_e$ covering a domain $\Omega_e$.
Hence, the error estimate in the energy is given by
\begin{equation}
|E_h - E| \leq \mathcal{C}\sum_{e}\left[h_e^{2k}|\ubarn|_{k+1,\Omega_e}^{2} + h_e^{2k}|\phibar|_{k+1,\Omega_e}^{2} + h_e^{2k+1}|\ubarn|_{k+1,\Omega_e}|\phibar|_{k+1,\Omega_e}\right]\,. \label{errorfinal}
\end{equation}

In the above equation and the analysis to follow, for simplicity, we restrict ourselves to the case where a single finite-element triangulation provides discretization for both root-electron-density and electrostatic potential. Finally, the error estimate to $O(h^{2k+1})$ is given by
\begin{equation}
|E_h - E| \leq \mathcal{C}\sum_{e}h_e^{2k}\left[|\ubarn|_{k+1,\Omega_e}^{2} + |\phibar|_{k+1,\Omega_e}^{2}\right]\,.\label{errorfinal1}
\end{equation}

\subsection{Optimal Coarse-Graining Rate}~\label{sec:optimal_mesh}

We now present the optimal mesh-size distribution that can be estimated by minimizing the approximation error in energy for a fixed number of elements. The approach here closely follows the treatment in~\cite{radio}.
Using the definition of the semi-norms, we rewrite equation~\eqref{errorfinal1} as
\begin{equation}
|E_h - E| \leq \mathcal{C} \sum_{e=1}^{N_{e}} \Bigl[ h_{e}^{2k}\int_{\Omega_{e}} \Bigl[|D^{k+1}\ubarn(\bx)|^{2} + |D^{k+1}\phibar(\bx)|^{2}\Bigr] \dx\Bigr]\,,
\end{equation}
where $N_e$ denotes the total number of elements in the finite-element triangulation. To obtain a continuous optimization problem rather than a discrete one, an element size distribution function $h(\bx)$ is introduced so that the target element size is defined at all points $\bx$ in $\Omega$, and we get
\begin{align}\label{errorN}
|E_h - E| &\leq  \mathcal{C} \sum_{e=1}^{N_{e}} \int_{\Omega_{e}}\Bigl[ h_{e}^{2k} \Bigl[|D^{k+1}\ubarn(\bx)|^{2} + |D^{k+1}\phibar(\bx)|^{2}\Bigr] \dx\Bigr]\,,\\
&\leq\mathcal{C'} \intomega h^{2k}(\bx)\Bigl[|D^{k+1}\ubarn(\bx)|^{2} + |D^{k+1}\phibar(\bx)|^{2}\Bigr] \dx\,.
\end{align}
Further, the number of elements in the mesh is in the order of
\begin{equation}
N_{e} \propto \intomega \frac{\dx}{h^{3}(\bx)} \label{elem}\,.
\end{equation}
The optimal mesh size distribution is then determined by the following variational problem which minimizes the approximation error in energy subject to a fixed number of elements:
\begin{equation}
\min_{h} \intomega \Bigl\{ h^{2k}(\bx)\Bigl[|D^{k+1}\ubarn(\bx)|^{2} + |D^{k+1}\phibar(\bx)|^{2}\Bigr] \Bigr\}\dx \quad \mbox{subject to}:\intomega\frac{\dx}{h^{3}(\bx)}=N_{e}\,.
\end{equation}
The Euler-Lagrange equation associated with the above problem is given by
\begin{equation}
2kh^{2k-1}(\bx)\Bigl[|D^{k+1}\ubarn(\bx)|^{2} + |D^{k+1}\phibar(\bx)|^{2}\Bigr] - \frac{3\eta}{h^{4}(\bx)} = 0\,,
\end{equation}
where $\eta$ is the Lagrange multiplier associated with the constraint. Thus, we obtain the following distribution
\begin{equation}\label{optimmesh}
h(\bx) = A \Bigl(|D^{k+1}\ubarn(\bx)|^{2} + |D^{k+1}\phibar(\bx)|^{2}\Bigr)^{-1/(2k+3)}\,,
\end{equation}
where the constant $A$ is computed from the constraint that the total number of elements in the finite element discretization is $N_e$.

The coarse-graining rate derived in equation~\eqref{optimmesh} has been employed to construct the finite-element meshes for different kinds of problems we study in the subsequent sections.

%% file: Numericalimplementation.tex
In this section, we turn to the numerical implementation of the variational formulation described in Section~\ref{formulation} using the finite-element basis. We first review the finite-element basis used in our study and subsequently discuss the solution procedure.

\subsection{Finite-Element Basis}
Traditionally, linear tetrahedral elements have been the elements of choice for many applications. These elements are well suited for problems involving complex geometries and requiring moderate levels of accuracy. However, in electronic structure calculations, geometries are simpler while the levels of accuracy desired are much higher. We investigate, in our study, if the desired chemical accuracy in electronic structure calculations can possibly be achieved more efficiently by using higher-order finite-element basis functions. We restrict our attention to a small, yet representative subset of existing finite-element basis functions. We employ in our study $C^0$ basis functions comprising of tetrahedral elements with interpolating polynomials up to degree two (TET4 and TET10) and hexahedral elements up to degree four (HEX8, HEX27, HEX64, HEX125). The number following the words `TET' and `HEX' denote the number of nodes in the element. We also investigate the use of spectral-elements, first introduced in~\cite{patera1984spectral}, and applied to electronic structure calculations in~\cite{batcho}. In conventional finite-elements, basis functions are constructed as Lagrange polynomials interpolated through equi-spaced nodes in an element, whereas,  spectral-elements use an optimal distribution of nodes corresponding to the roots of derivatives of Chebyshev or Legendre polynomials. Such a distribution does not have nodes on the boundaries of an element, and it is common to append nodes on element boundaries which guarantees $C^{0}$ basis functions. These set of nodes are commonly referred to as Gauss-Lobatto-Chebyshev or Gauss-Lobatto-Legendre points. We further note that, for a high order of interpolation, conventional finite-elements result in an ill-conditioned problem, whereas, spectral-elements are devoid of this deficiency~\cite{boyd2001chebyshev}. In this work, we use spectral-elements with basis functions of third and fourth degree having nodes positioned at the Gauss-Lobatto-Legendre points.

In the studies to follow, we use Gauss quadrature rules to numerically evaluate the integrals arising in the formulation. We choose our numerical quadrature rules such that the error incurred through quadratures is higher order in comparison to discretization errors. Further, we restrict ourselves to a single finite-element discretization for all solution fields.  We next present a brief description of the solution procedure and solvers employed in this work.

\subsection{Solution Procedure}
The analysis in appendix~\ref{solve} shows that the incremental problem~\eqref{stagg} corresponding to a single step in Newton method is uniquely solvable in a staggered approach, and hence we employ the staggered solution procedure in our numerical studies. In order to provide a good initial guess for the non-linear solvers based on quasi-Newton methods, we first use a non-linear conjugate gradient method implemented using Polak-Ribere scheme~\cite{shewchuk1994} to reduce the residuals corresponding to root-electron-density. This solution is subsequently used as an initial guess to the inexact Newton solver provided by the KINSOL pacakage~\cite{collier2006} that rapidly reduces the residuals in the vicinity of the solution providing quadratic convergence. This is a Newton's Method with Goldstein-Armijo line searches~\cite{dennis1996}. The generalized-minimal residual method (GMRES \cite{saad1986gmres}) is used as the associated linear solver used in the inexact Newton solver, while bi-conjugate gradient stabilized method (BICG, \cite{van1992bi}) and transpose-free quasi-minimal residual method (TFQMR, \cite{freund1993}) are alternate options for linear solvers. All these algorithms are based on Krylov subspace methods and hence require the evaluation of Hessian matrix only through its product with a vector, which is approximated by directional difference quotients. The solution of electrostatic potential (Poisson equation) and kernel potentials (Helmholtz equations) in the staggered procedure is computed using preconditioned linear conjugate gradient method or preconditioned GMRES. A suite of different preconditioners provided by the HYPRE \cite{chowhypre} package is used to improve the conditioning of the problem and  computational efficiency.

%% file: NumericalResults.tex
\subsection{Rates of Convergence}
We now study the convergence rates of the finite-element approximation with decreasing mesh sizes for different polynomial orders of interpolation. We use three benchmark problems in this study, which include: (i) Hydrogen atom that represents a linear problem with a singular nuclear potential; (ii) Helium atom that represents a non-linear problem with a singular nuclear potential; (iii) Aluminum atom using pseudopotentials that represents a non-linear problem with a smooth external pseudopotential. In the studies which do not use pseudopotentials, the nuclear charges are treated as point charges located on the nodes of the finite-element triangulation, and the discretization provides a regularization for the nuclear charges. We note that the self-energy of the nuclei in this case is mesh-dependent and diverges upon mesh refinement. Thus, the self energy is also computed on the same mesh that is used to compute the total electrostatic potential, which ensures that the divergent components of the variational problem on the right hand side of equation~\eqref{elReformulation} and the self energy exactly cancel owing to the linearity of the Poisson equation.

We adopt the following procedure for conducting the convergence study. The coarsest mesh having $N_e$ elements is constructed with a coarsening rate determined by equation~\eqref{optimmesh} using the {\it a priori} knowledge of the far-field asymptotic solutions of electronic fields. We note that as far-field asymptotic solutions are used in \eqref{optimmesh} to compute the coarse-graining rates as opposed to the exact solutions, the obtained meshes are only optimal in the far-field and possibly sub-optimal near the nuclei. Nevertheless, the meshes constructed using this approach still provide an efficient way of resolving the vacuum in non-periodic calculations as opposed to using a uniform discretization or employing an ad-hoc coarse-graining rate. The sequence of refined meshes are subsequently constructed by a uniform subdivison of the coarse-mesh, which represents a systematic refinement of the approximation space. The finite-element ground-state energies, $E_h$, obtained from each of these subdivisions for the highest order element(HEX125) are used to fit an expression of the form
\begin{equation}
|E_0 - E_h| =  \mathcal{C} (1/N_e)^{2k/3}.
\end{equation}
to determine the constants $E_0$, $\mathcal{C}$ and $k$. All convergence study plots for different orders of finite-elements in the subsequent sub-sections show the relative error $\frac{|E_0-E_h|}{|E_0|}$ plotted against $(1/N_e)^{1/3}$ for the value of $E_0$ obtained using the HEX125 element. The slopes of these curves provide the rate of convergence of the finite-element approximation error in energy for the system being studied.

\subsubsection{Hydrogen Atom}
We first begin with the simplest system to study, namely the Hydrogen atom. Note that the orbital-free DFT functional in the case of a single-electron system takes a simplified form due to the absence of electron-electron interactions. Further, the kinetic energy of the single-electron system can be represented exactly by the von-Weizsacker kinetic energy functional with $\lambda = 1$, while the electrostatic energy is solely comprised of nuclear-electron interactions. An analytical solution for the root-electron-density for Hydrogen atom is known, and the radial solution is given by
\begin{equation}
\bar{u}(r) = \sqrt{\frac{1}{\pi}}\exp{(-r)}\,.
\end{equation}
In this case, using the theoretical solution for $\bar{u}$, the optimal mesh coarsening rate as determined by equation~\eqref{optimmesh} is given by the following expression
\begin{equation}\label{hyd_optim}
h(r) = A\exp\left(\frac{2\,r}{2k+3}\right)\,.
\end{equation}
Using this coarse-graining rate, we perform the numerical convergence study with both tetrahedral and hexahedral elements and the results are shown in figure~\ref{fig:SingleHydConvgRate}. In the present case, since the ground-state energy of Hydrogen is known theoretically, this theoretical value of -0.5 Hartree is used for $E_0$ in computing the relative errors in energy.
 	\begin{center}
		\includegraphics[width=0.6\textwidth]{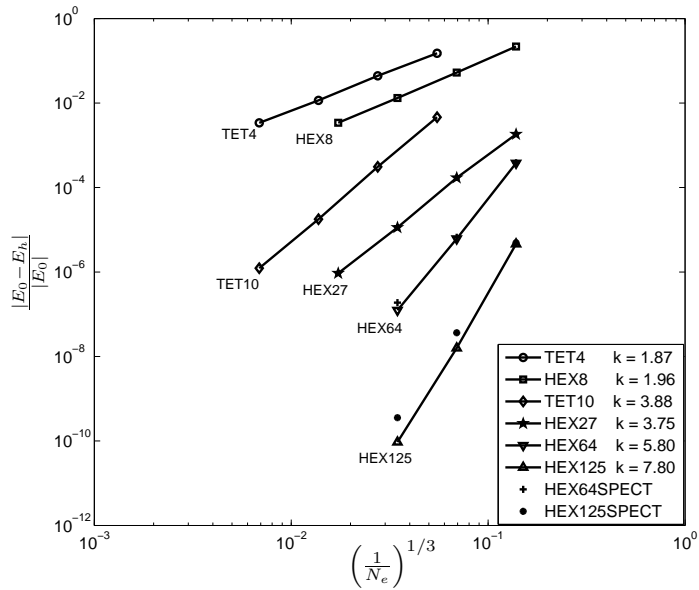}
		\captionof{figure}{Convergence rates for the finite-element approximation of a single Hydrogen atom using orbital-free DFT.}
		\label{fig:SingleHydConvgRate}
	\end{center}
Results show that all the elements have close to optimal rate of convergence, $O(h^{2k})$, where $k$ is the degree of the polynomial. It is interesting to note that optimal rates of convergence are obtained although the nuclear potential approaches a singular solution (Coulomb potential) upon refinement. We note that, while both hexahedral and tetrahedral elements with same interpolation orders converge in error at similar rates, hexahedral elements provide better accuracies for a given mesh size (HEX8 and HEX27 elements are more accurate than TET4 and TET10 elements respectively). Secondly, as the relative accuracy of calculations increase, elements with linear interpolation soon become untenable due to the large increase in number of elements for a small drop in error. For instance, the computation of energy of a single Hydrogen atom to a relative error of $10^{-2}$ required over a million linear TET4 elements in comparison to just a few thousand elements with quadratic TET10 elements. Further, relative errors up to $10^{-4}$ are achieved with just few hundreds of HEX64 and HEX125 elements, and even higher accuracies are achieved with few thousands of these elements.

\subsubsection{Helium Atom}
The next example we study is a single Helium atom, the simplest system displaying the full complexity of an all-electron calculation. The kinetic energy is treated using the TFW functional with $\lambda=\frac{1}{5}$, and the exchange-correlation effects are treated using the LDA approximation. In order to determine the mesh coarse-graining rate for this problem, we note that the asymptotic solution for the electron-density in the TFW model is governed by an exponential decay from the nucleus~\cite{parr} and is given by
\begin{equation}
\rho(r) \sim \frac{1}{r^2}\exp\Bigl[-\sqrt{\frac{8\,\bar{\mu}}{\lambda}} \;r\Bigr]\;\;\;\;\text{for}\;\;\; r\rightarrow\infty\,,
\end{equation}
where $\bar{\mu}$ is the lagrange multiplier associated with the constraint in the orbital-free DFT problem. Since the decay of electron-density is dominantly influenced by the exponential term, the root-electron-density $\bar{u}(r)$ for the purpose of determining the mesh coarse-graining rate is assumed to be
\begin{equation}
\bar{u}(r) = \sqrt{\frac{2\xi^3}{\pi}} \exp(-\xi\,r) \;\;\;\;\text{where}\;\;\;\;\xi = \frac{1}{2}\sqrt{\frac{8\,\bar{\mu}}{\lambda}}\,.
\end{equation}
The electrostatic potential is governed by the Poisson equation with total charge density being equal to the sum of $\bar{u}^{2}(r)$ and $2\delta(r)$ and is given by
\begin{equation}
\bar{\phi}(r) = -2 \exp{(-2\xi \,r)}\left(\xi + \frac{1}{r}\right)\,.
\end{equation}
Using the above equations, the mesh coarse-graining rate from equation~\eqref{optimmesh}  is given by
\begin{equation}\label{hel_optim}
h(r) = A\left[\frac{2}{\pi}\xi^{2k+5} \exp{(-2\,\xi\,r)} + 4 \exp{(-4\,\xi\,r)}\left[\xi^{k+2} 2^{k+1} + \sum_{n=0}^{k+1} \dbinom{k+1}{n} \frac{2^{n} \xi^{n} (k+1-n)!}{r^{k-n+2}}\right]^{2}\right]^{-1/(2k+3)}\,.
\end{equation}
Note that $\bar{\mu}$ which appears in the above equation is a parameter not known {\it a priori}. Hence, we find the value of $\bar{\mu}_h$ on a coarse mesh initially and then use in the above equation to obtain $h(r)$. Using this coarse-graining rate, we perform the numerical convergence study with both tetrahedral and hexahedral elements as explained before.
Figure~\ref{fig:SingleHeliumConvgRate} shows the convergence results for the various elements. As in the case of Hydrogen atom, we obtain close to optimal convergence rates although the governing equations are non-linear in nature and the nuclear potential approaches a singular solution upon refinement. Recent mathematical analysis~\cite{ortner} shows that the finite-element approximation for the orbital-free DFT problem does provide optimal rates of convergence, however, this analysis assumes that the nuclear potentials are smooth. To the best of our knowledge mathematical analysis of the finite-element approximation of the orbital-free DFT problem with singular nuclear potentials is still an open question. The present numerical investigation provides evidence that optimal rates of convergence are achieved for the non-linear orbital-free DFT problem even with singular potentials.
\begin{center}
\includegraphics[width=0.6\textwidth]{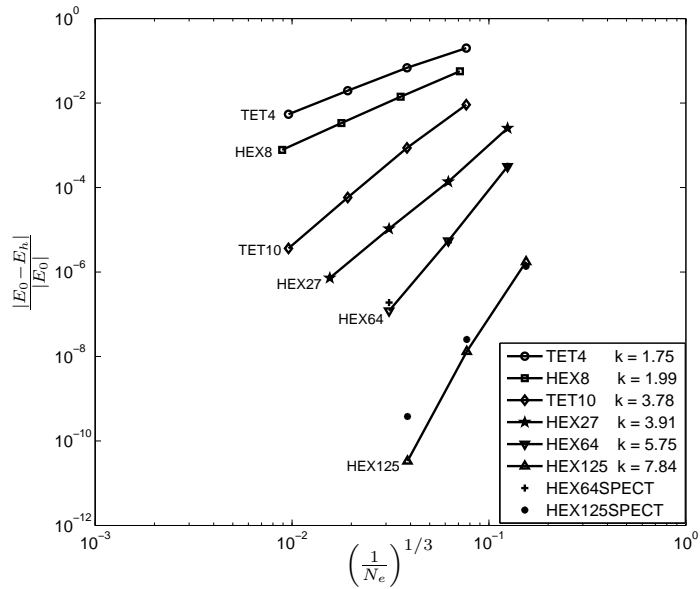}
\captionof{figure}{Convergence rates for the finite-element approximation of a single Helium atom using orbital-free DFT.}
\label{fig:SingleHeliumConvgRate}
\end{center}

\subsubsection{Pseudopotential Calculations}
We now turn to  pseudopotential calculations in multi-electron systems. A pseudopotential constitutes the effective potential of the nucleus and core electrons experienced by the valence electrons. In solid state physics, since the valence electrons are mainly responsible for chemical bonding and the chemical and physical properties of solids, it is common to compute material properties under the pseudopotential approximation. Further, as the electron-density corresponding to the valence electrons is closer to a free-electron gas, the model kinetic energy functionals in orbital-free DFT are better approximations for pseudopotential computations in comparison to all-electron calculations.

\paragraph* {Single Atom:}
The first example in this category is that of a single Aluminum atom. Here, we use the TFW kinetic energy functionals with $\lambda=\frac{1}{5}$, the LDA approximation for exchange correlation functional, and  the Goodwin-Needs-Heine pseudopotential~\cite{goodwin1990}.
The mesh coarsening rate is derived in the same way as described in the case of Helium atom. Since the decay of the electron density is dominantly influenced by the exponential term, the root-electron-density $\bar{u}(r)$ for this problem is assumed to be of the form
\begin{equation}
\bar{u}(r) = \sqrt{\frac{3\xi^3}{\pi}} \exp(-\xi\,r)\,,\;\;\;\;\text{where}\;\;\;\;\xi = \frac{1}{2}\sqrt{\frac{8\,\bar{\mu}}{\lambda}}\,.
\end{equation}
The electrostatic potential is governed by the Poisson equation with total charge being equal to the sum of  $\bar{u}^{2}(r)$ and $3\delta(r)$ and is given by
\begin{equation}
\bar{\phi}(r) = -3 \exp{(-2\xi \,r)}\left(\xi + \frac{1}{r}\right)\,.
\end{equation}
Using the above equations, the mesh coarse-graining rate from equation~\eqref{optimmesh} is given by
\begin{equation}\label{al_optim}
h(r) = A\left[\frac{3}{\pi}\xi^{2k+5} \exp{(-2\,\xi\,r)} + 9 \exp{(-4\,\xi\,r)}\left[\xi^{k+2} 2^{k+1} + \sum_{n=0}^{k+1} \dbinom{k+1}{n} \frac{2^{n} \xi^{n} (k+1-n)!}{r^{k-n+2}}\right]^{2}\right]^{-1/(2k+3)}\,.
\end{equation}
Figure~\ref{fig:SingleAlumConvgRate} shows the rates of convergence for the various elements considered, which exhibit close to optimal rates of convergence. The main observation that distinguishes this study from the all-electron study (Helium and Hydrogen atoms) is that all orders of interpolation provide much greater accuracies for pseudopotential calculations for the same number of elements. Linear basis functions are able to approximate the energy to two orders of greater accuracy. Ground state energies within relative errors of $10^{-6}$ can be achieved even with quadratic elements like TET10 and HEX27. The reason for this improved accuracy can be attributed to the smooth pseudopotential field in contrast to the singular nuclear potential field describing interactions between the nuclei and electrons.

	\begin{center}
	\includegraphics[width=0.6\textwidth]{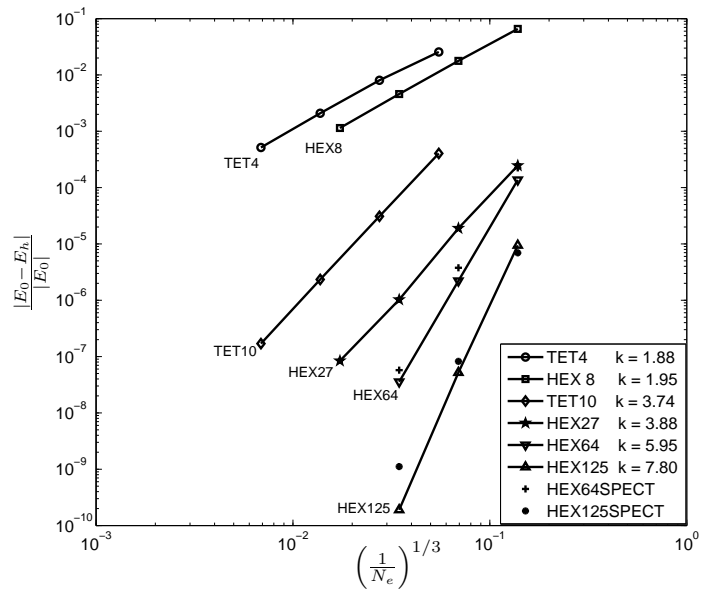}
	\captionof{figure}{Convergence rates for the finite-element approximation of a single Aluminum atom using orbital-free DFT.}
	\label{fig:SingleAlumConvgRate}
	\end{center}

\paragraph*{Perfect crystal with periodic boundary conditions:}
The next example considered is that of a perfect Aluminium crystal. Periodic boundary conditions are imposed on the root-electron-density and the electrostatic potential on a single face-centered cubic unit cell. The kinetic energy functional is modeled using kernel energies in this case. The details of the local reformulation with kernel energies have been discussed in Section \ref{formulation}. Figure \ref{fig:DIAlumConvgRate} shows the rate of convergence of the energy error using the density independent (DI) kernel functionals. We note that the convergence trend is similar to that of a single atom calculation with linear basis functions approximating the energy to a better accuracy in comparison to an all-electron calculation for the same number of elements. Simulations have also been performed to demonstrate the optimal convergence rates using density dependent (DD) kernels and are shown in figure~\ref{fig:DDAlumConvgRate}.

\begin{center}
	 \includegraphics[width=0.6\textwidth]{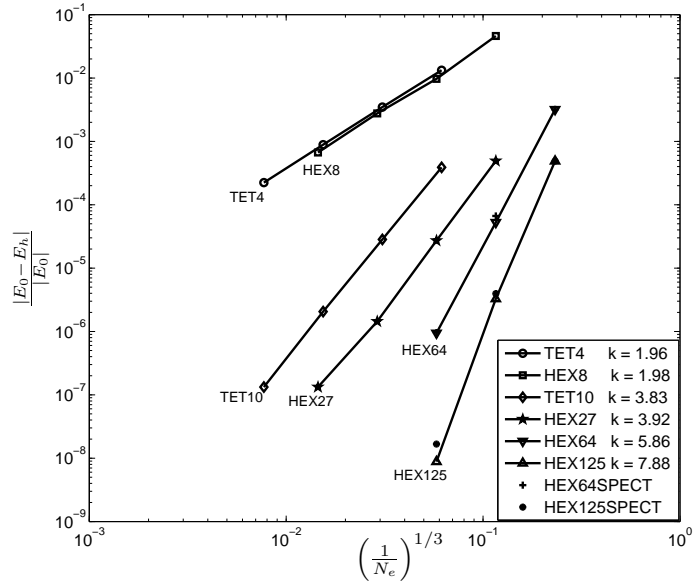}
	\captionof{figure}{Convergence rates for the finite-element approximation of bulk Aluminum using orbital-free DFT with DI kernel energy.}
	\label{fig:DIAlumConvgRate}
\end{center}

\begin{center}
	 \includegraphics[width=0.6\textwidth]{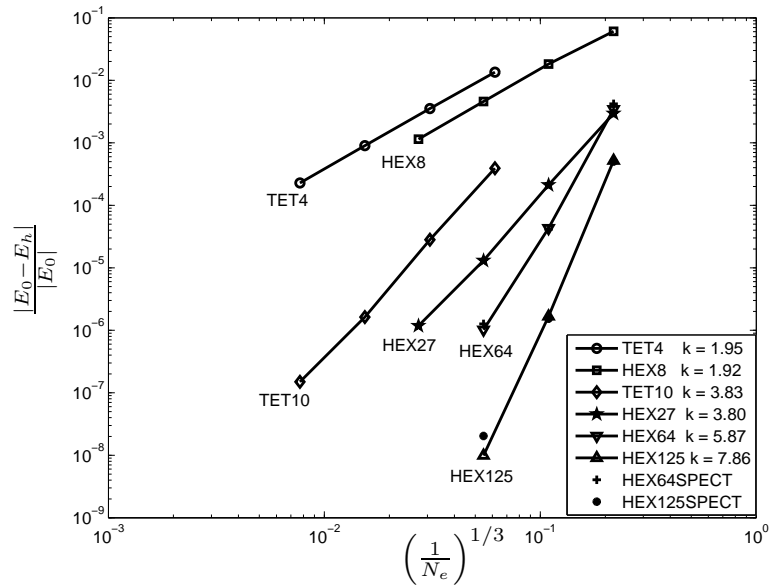}
	\captionof{figure}{Convergence rates for the finite-element approximation of bulk Aluminum using orbital-free DFT with DD kernel energy.}
	\label{fig:DDAlumConvgRate}
\end{center}

\subsection{Computational Cost}
We now turn towards studying the computational efficiency afforded by the use of higher-order finite-element approximations in the orbital-free DFT problem. As seen from the convergence studies, higher-order finite-element approximations result in a faster rate of convergence. However, the need for higher-order quadratures for numerical approximation of the integrals in the finite-element formulation increases the per element cost of computations in using higher-order elements. Further, the increase in the bandwidth of the Hessian matrix affects the performance of the iterative solvers. In order to unambiguously determine the computational efficiency of higher-order elements, we measure the CPU time taken for the simulations conducted on benchmark problems for a wide range of meshes providing different relative accuracies. We also conduct this investigation on larger problems comprising of Aluminum clusters with varying sizes. All the simulations are conducted using meshes with the coarse-graining rates determined by equation~\eqref{optimmesh}. All the numerical simulations reported in this work are conducted using a parallel implementation of the code based on MPI and were executed on a parallel computing cluster with following specifications: (2) dual-core AMD Opteron2218 2.6GHz processors per node, 8 GB memory per node and 1Gbps ethernet connectivity between nodes. The various single atom calculations were executed using 1 to 8 processors (cores), while the results for the larger aluminium clusters discussed subsequently were executed on 40 processors. It was verified that code scales linearly on this parallel computing architecture and hence the times reported here represent the total CPU time for the calculation, which is equivalent to the wall-clock time on a single processor. \cn

\subsubsection{Single Atom Calculations}\label{sec:comp_atom}
A series of meshes for each element type are constructed using the optimal mesh coarse-graining rates for the benchmark systems considered (cf. equations~\eqref{hyd_optim},\eqref{hel_optim} and \eqref{al_optim}). The number of elements are varied to obtain finite-element approximations with varying accuracies that target relative energy errors between $10^{-1}-10^{-6}$. These meshes are used to solve the orbital-free DFT problem for the benchmark problems comprising of Hydrogen, Helium and Aluminum atoms. For the various finite-element interpolations considered in the present study, the relative energy error is plotted against the total CPU time for simulations on the series of meshes constructed. The CPU times are normalized with the longest time taken in the series of simulations for a given material system. Figures~\ref{fig:SingleHydTime}, \ref{fig:SingleHeliumTime} and \ref{fig:SingleAlumTime} show these results for Hydrogen, Helium and Aluminum atom systems respectively.

Our results indicate that higher-order interpolations become computationally more efficient as the desired accuracy of the computations increase, in particular for relative errors of the order of $10^{-4}$ and lower, which corresponds to chemical accuracy. Even at moderate accuracies, HEX8 elements are ten-fold computationally more efficient than TET4 elements. For relative errors of $10^{-2}-10^{-3}$, quadratic HEX27 elements perform similar to cubic HEX64 elements, whereas quartic HEX125 elements take relatively longer times clearly showing the trade-off between accuracy and computational cost. Nevertheless, all higher-order elements are computationally more efficient than linear TET4 elements. As, we examine the relative errors lower than $5\times 10^{-5}$, quartic HEX125 elements offer a ten-fold computational savings in comparison to cubic HEX64 elements, while other low-order elements start deteriorating in terms of performance. An interesting observation is that the computational advantage offered by the spectral-elements is at least two-fold in the case of quartic HEX125Spectral element in comparison with quartic HEX125 element. We attribute the improved performance of the spectral-elements to the better conditioning of the discrete system of equations. Comparing the results across different materials systems, we observe that the performance of lower-order elements is far worse in the case of Hydrogen and Helium atom systems in comparison to Aluminum atom system. For instance, at a relative error of $10^{-3}$, the solution time using TET4 is more than three orders of magnitude larger than HEX125Spectral elements for the case of Helium atom. However, the solution time is only two orders of magnitude larger for TET4 over HEX125Spectral elements for Aluminium atom.

In summary, we note the following major observations from this investigation. Firstly, for a given level of accuracy, hexahedral elements converge faster than tetrahedral elements with the same order of interpolation. Secondly, for chemical accuracy corresponding to relative errors lower than $5\times 10^{-5}$, the computational efficiency improves significantly with the order of the element. Moreover, spectral-elements are computationally more efficient in comparison to the conventional Lagrange elements of the same order. As further evidence, we note that the simulations on Hydrogen and Helium atoms show a staggering 100-1000 fold computational savings that are realized with the use of HEX125Spectral elements providing a relative error of $10^{-4}$ in comparison to linear TET4 elements resulting in a much higher relative error of around $10^{-2}$. Lastly, qualitatively speaking, the sequence of graphs of relative error vs. normalized CPU time for the various elements tend towards increasing accuracy and computational efficiency with increasing order of finite-element interpolation. We also note that the point of diminishing returns in terms of computational efficiency of higher-order elements for relative errors commensurate with chemical accuracy is beyond the fourth-order, and is yet to be determined.
	
   \begin{center}
	\includegraphics[width=0.65\textwidth]{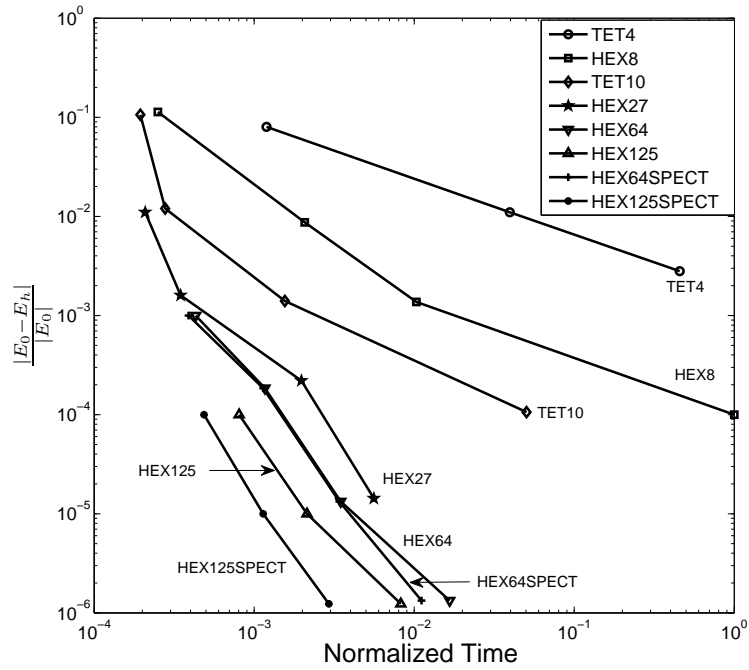}
	\captionof{figure}{Computational efficiency of various orders of finite-element approximations. Case study: Hydrogen atom.}
	\label{fig:SingleHydTime}
	\end{center}	

    \begin{center}
	\includegraphics[width=0.65\textwidth]{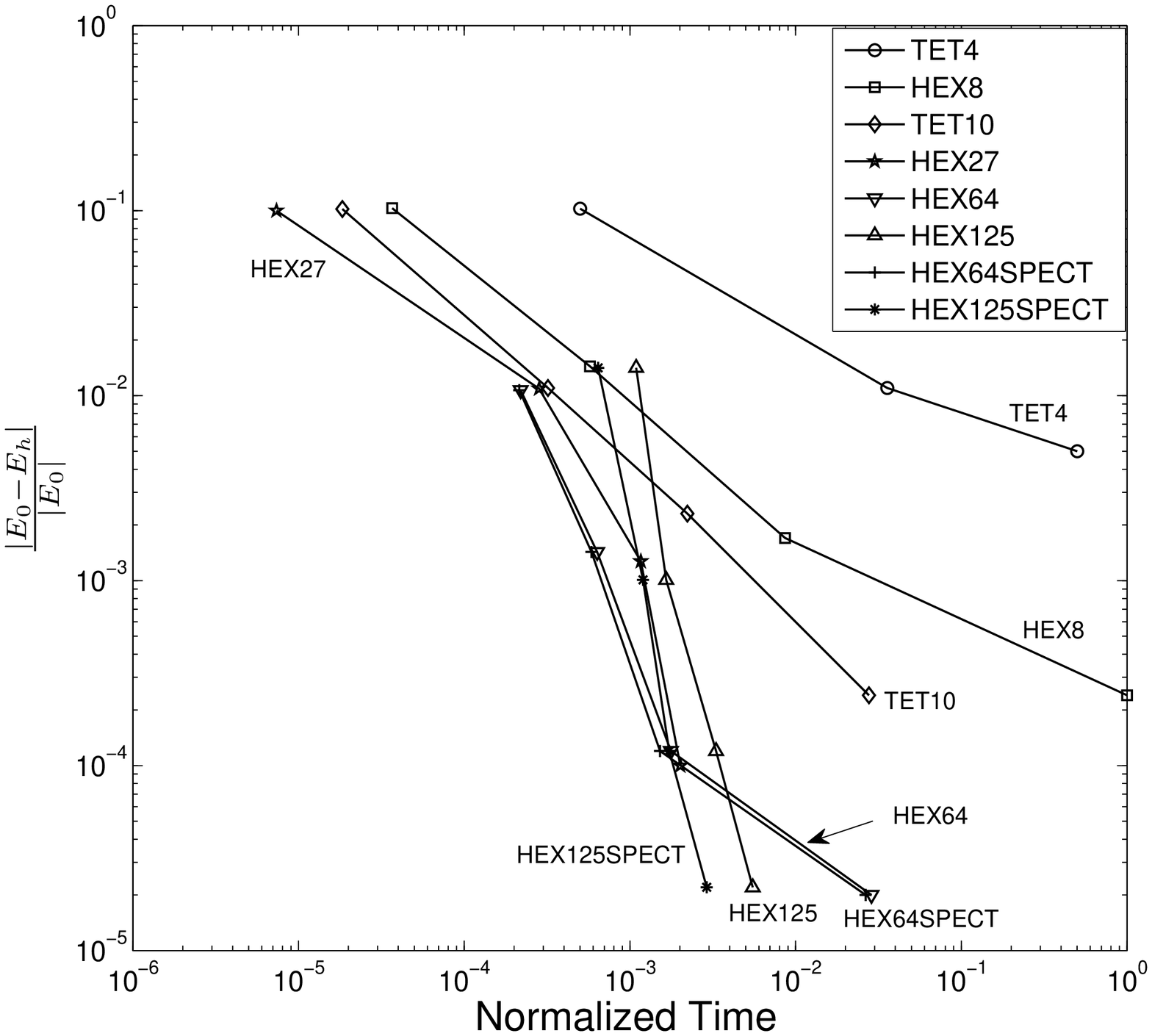}
	\captionof{figure}{Computational efficiency of various orders of finite-element approximations. Case study: Helium atom.}
	\label{fig:SingleHeliumTime}
	\end{center}

	\begin{center}
	\includegraphics[width=0.65\textwidth]{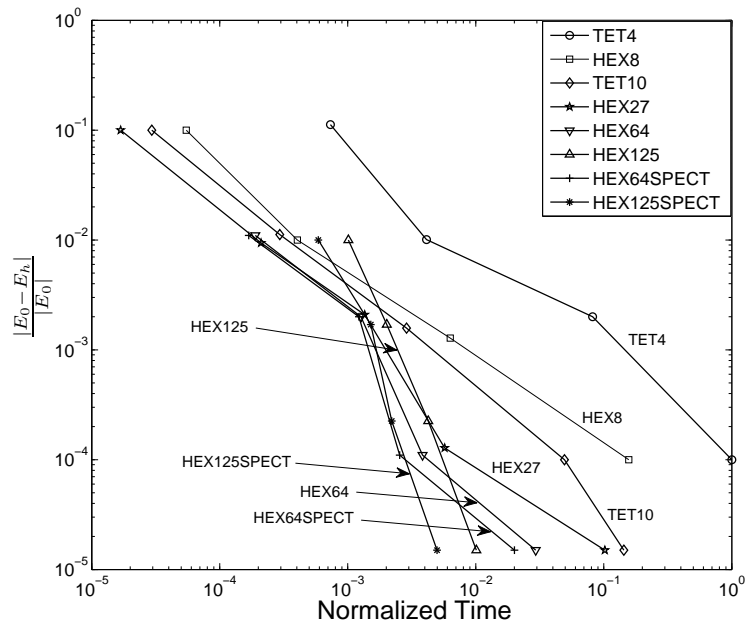}
	\captionof{figure}{Computational efficiency of various orders of finite-element approximations. Case study: Aluminum atom.}
	\label{fig:SingleAlumTime}
	\end{center}

\subsubsection{Large Aluminum Clusters}\label{AlClusters}
In order to further investigate the computational efficiency afforded by higher-order elements in orbital-free DFT, we consider larger materials systems comprising of Aluminum clusters containing 1 x 1 x 1,  3 x 3 x 3,  5 x 5 x 5  face-centered cubic unit cells. We restrict this investigation to studying the computational efficiency of linear TET4 and quartic HEX125Spectral elements, which constitute the extreme ends of the spectrum of elements studied in this work. The finite-element meshes for these clusters were constructed to be uniform in the cluster region, while coarse-graining away from the cluster region. The mesh coarsening rate in the vacuum is determined by using in equation~\eqref{optimmesh} the asymptotic solution of the far-field electronic fields, estimated as a superposition of single atom far-field asymptotic fields given by equation~\eqref{al_optim}. Figure~\ref{fig:ClusterAlumTime} shows the relative errors vs. the normalized CPU time for the various simulations conducted. The reference time taken here is that of a 3x3x3 aluminium cluster simulated using TET4 elements which took about 3000 CPU hours (approximately 75 hours on 40 AMD Opteron 2.6 GHz processors), and the relative error in the energy for this simulation is of the order $10^{-2}$. As can be observed from the results, quartic HEX125Spectral element is over hundred-fold more computationally efficient and provides at least an order of magnitude greater accuracy than the linear TET4 element. Further, it is interesting to note that the simulation performed on a 5x5x5 aluminium cluster (666 atoms) using quartic HEX125Spectral elements takes only marginally more time than a 1x1x1 cluster (14 atoms) discretized with linear TET4 elements, while providing two orders-of-magnitude greater accuracy. These results demonstrate that the computational efficiency of higher-order elements observed in single atomic systems in section~\ref{sec:comp_atom} also holds for larger materials systems with varying sizes, and corroborates the advantage of using higher-order spectral-elements in orbital-free DFT calculations.
\begin{center}
\includegraphics[width=0.7\textwidth]{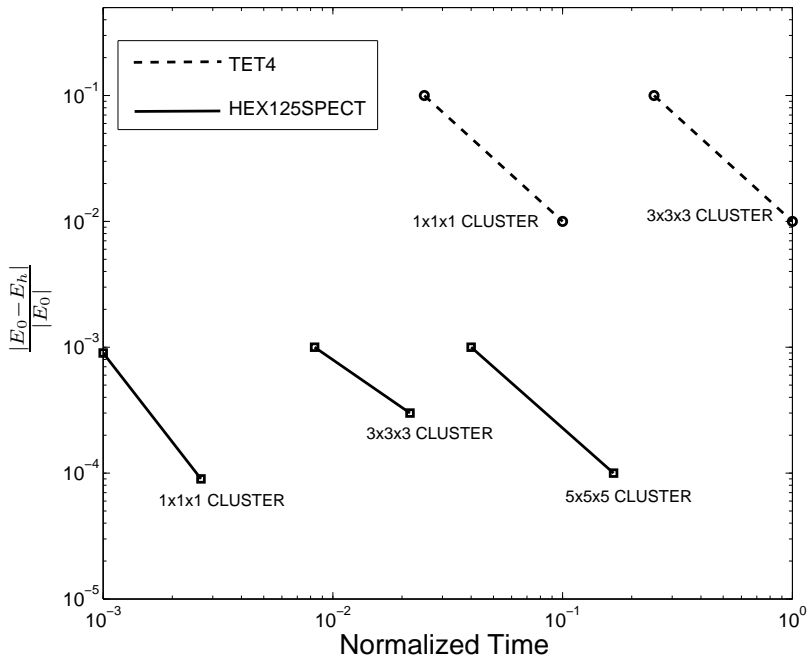}
\captionof{figure}{Computational efficiency of various orders of finite-element approximations. Case study: Aluminum clusters.}
\label{fig:ClusterAlumTime}
\end{center}

%% file: conclusions.tex
In the present study, we conduct a numerical analysis of the finite-element discretization of the orbital-free DFT problem with particular focus on evaluating the computational efficiency afforded by the use of higher-order finite-element discretizations. We used the saddle-point formulation of the orbital-free DFT problem proposed in~\cite{ofdft} as our starting point. Since a higher-order finite-element discretization of the saddle-point problem allows for the possibility of using different orders of interpolation for various electronic fields, namely the potential fields and the root-electron-density, we investigated robustness of two possible solution schemes: (i) a simultaneous scheme where all the electronic fields are solved concurrently; (ii) a staggered approach, where the root-electron-density is solved using a quasi-Newton based iterative solver, while the potential fields are computed consistently for every trial root-electron-density. Our analysis of the solvability conditions indicated that, while the discrete system of equations in the staggered approach is invertible for any choice of interpolations of the electronic fields, the discrete equations in simultaneous scheme are not solvable if the dimension of the vector-space approximating root-electron-density far exceeds that of electrostatic potential and the Hessian matrix corresponding to root-electron-density is rank deficient to a certain degree which depends on the initial guess. This non-invertibility of the discrete set of equations corresponding to the simultaneous scheme is an inherent deficiency of saddle-point problems which cannot be surpassed. Thus, while the staggered scheme is robust, the simultaneous scheme is sensitive to the initial guess, an aspect which was also observed in our numerical simulations.

In order to aid our studies on the computational efficiency of higher-order finite-elements in orbital-free DFT, we developed error estimates for the approximation error in energy in terms of the ground-state electronic fields and characteristic mesh-size. Using these error estimates and the {\it a priori} knowledge of the asymptotic solutions of far-field electronic fields, we constructed mesh coarse-graining rates for the various benchmark problems, which include single-atom systems comprising of Hydrogen, Helium and Aluminum atoms. We first investigated the performance of higher-order elements by studying the convergence rates of various elements up to fourth-order including tetrahedral elements, hexahedral elements and spectral finite-elements. In all cases, we observed close to optimal convergence rates. It is indeed worthwhile to note that optimal convergence is obtained although the orbital-free DFT problem is non-linear and the potential fields in Hydrogen and Helium atoms are singular. Having demonstrated the optimal convergence, we subsequently investigated the computational efficiency afforded by the use of higher-order elements. To this end, we used mesh coarse-graining rates determined from the proposed mesh adaption scheme and studied the CPU time required to solve the benchmark problems as well as larger systems containing up to 666 atom Aluminum clusters. Our results demonstrate that significant computational savings can be realized by using higher-order elements. For instance, a staggering 100-1000 fold savings in terms of CPU time and two orders of magnitude better accuracy are realized by using fourth-order hexahedral spectral-elements in comparison to linear tetrahedral elements.

The prospect of using higher-order finite-elements as basis functions for electronic structure calculations is indeed promising. While finite-elements have the advantages of handling complex geometries and boundary conditions and scale well on massively parallel computing architectures, their use has been limited in electronic structure calculations as they compare unfavorably to plane-wave and atomic-orbital basis functions due to the large number of basis functions required to reach chemical accuracy. The present study shows that the use of higher-order discretizations can alleviate this problem, and presents a useful direction for electronic structure calculations using finite-element discretization. As part of the ongoing work, we are currently investigating the efficiency of these higher-order discretizations for the solution of Kohn-Sham DFT equations, and their performance in comparison to plane-wave and atomic-orbital basis functions. Further, the implications of using higher-order finite-element approximations in the recently proposed quasi-continuum reduction of orbital-free DFT~\cite{qcofdft} is a worthwhile subject for future investigation.

%% file: acknowledgements.tex
We thank Kenneth Leiter (U.S. Army Research Labs) for assistance in improving the software architecture of the code that is used to conduct our simulations. We gratefully acknowledge the support of Air Force Office of Scientific  Research under Grant No. FA9550-09-1-0240 and Army Research Office under Grant No. W911NF-09-0292.

%% file: appendix.tex
\section{Solvability of the Discrete FE Problem}\label{solve}
For the sake of clarity and simplicity we first present the solvability conditions for the discrete problem~\eqref{fem}. Similar arguments can be used to extend the analysis in the case of kernel potentials (equation~\eqref{fem_kernel}).
\input{solve.tex}

\section{Estimate of Energy Error with Kernels}\label{error_kernel}
Let ($\ubarn_h$, $\phibar_h$, $\mubar_h$, $\alpj_h$, $\betj_h$) and ($\ubarn$, $\phibar$, $\mubar$, $\alpj$, $\betj$) be the solutions of the discrete finite element problem~\eqref{fem_kernel} and the continuous problem~\eqref{infsup_kernel} respectively for a given set of nuclear positions.  Kernel functionals are valid for periodic systems where ground-state electron-density $\ubarn$ is a perturbation of uniform electron gas. Hence, we assume $\ubarn$  is bounded from above and below in subsequent analysis.  The ground state energy in the discrete and the continuous formulations can be expressed as
\begin{equation}
\begin{split}
&E_h(\ubarn_h,\phibar_h,\alpj_h,\betj_h) =  \intomega\Bigl[\frac{1}{2}|\del \ubarn_h|^{2} \;+\; F(\ubarn_h) \; - \; \frac{1}{8\pi} |\del \phibar_h|^2  +  (\ubarn_h^2 \; +\; b)\phibar_h\Bigr] \dx\\
&+ \sum_{J=1}^{m}\Bigl\{\intomega \Bigl[\frac{1}{A_J\,B_J}\del\alpj_h\cdot\del\betj_h + \frac{1}{A_J}\alpj_h\betj_h  + \betj_h \,\ualp_h  +  \alpj_h\, \ubet_h  +  A_{J}\, \ualpbet_h\Bigr] \dx\Bigr\},
\end{split}
\end{equation}
and
\begin{equation}\label{contengyker}
\begin{split}
&E(\ubarn,\phibar,\alpj,\betj) =  \intomega\Bigl[\frac{1}{2}|\del \ubarn|^{2} \;+\; F(\ubarn) \; - \; \frac{1}{8\pi} |\del \phibar|^2  +  (\ubarn^2 \; +\; b)\phibar\Bigr] \dx\\
&+ \sum_{J=1}^{m}\Bigl\{\intomega \Bigl[\frac{1}{A_J\,B_J}\del\alpj\cdot\del\betj + \frac{1}{A_J}\alpj\betj  + \betj \,\ualp  +  \alpj\, \ubet  +  A_{J}\, \ualpbet\Bigr] \dx\Bigr\},
\end{split}
\end{equation}
where
\begin{equation*}
F(u) = C_F\;u^{10/3} +  \varepsilon_{xc}(u^2)u^2.
\end{equation*}
We begin by expanding $E_h(\ubarn_h,\phibar_h,\alpj_h,\betj_h)$ about the solution of the continuous problem, i.e $\ubarn_h = \ubarn + \delta u$, $\phibar_h = \phibar + \delta \phi$, $\alpj_h = \alpj + \delta \alpjn $, $\betj_h = \betj + \delta \betjn$.  Using Taylor series expansion, we get
\begin{equation}\label{taylor_ker}
\begin{split}
&E_h(\ubarn_h,\phibar_h,\alpj_h,\betj_h) =  \\
&\intomega \Bigl\{\frac{1}{2}|\del(\ubarn + \delta u)|^{2}  +  F(\ubarn + \delta u)
- \;\frac{1}{8\pi} |\del (\phibar + \delta \phi)|^{2}  +  \left[(\ubarn + \delta u)^2 +b \right]\left[\phibar + \delta \phi\right] \\
&+ \sum_{J=1}^{m}\Bigl\{ \Bigl[\frac{1}{A_J\,B_J}\left[\del(\alpj + \delta \alpjn)\cdot\del(\betj + \delta \betjn)\right] + \frac{1}{A_J}(\alpj + \delta \alpjn)(\betj + \delta \betjn) \\
 &+ (\betj + \delta \betjn)\,(\ubarn + \delta u)^{2\alpha}  +  (\alpj + \delta \alpjn)\,(\ubarn + \delta u ) ^{2\beta} +  A_{J}\, (\ubarn + \delta u)^{2(\alpha + \beta)}\Bigr]\Bigr\}\Bigr\}\dx.
 \end{split}
 \end{equation}
 Since ($\ubarn$, $\phibar$, $\mubar$, $\alpj$, $\betj$) satisfy the Euler-Lagrange equations of the functional  ~\eqref{infsup_kernel}, we have
\begin{equation*}\label{eulerker}
\begin{split}
&\intomega \Bigl\{\del \ubarn\cdot  \del \delta u +  F'(\ubarn) \delta u  + 2 \, \ubarn\; \delta u\; \phibar + \sum_{j=1}^{m} \Bigl[2\,\alpha\,\ubarn^{2\alpha - 1} \,\delta u\,\betj + 2\,\beta\,\ubarn^{2\beta - 1}\, \delta u\,\alpj \\
&+  2\,(\alpha+\beta)\,\ubarn^{2(\alpha+\beta) - 1}\, \delta u\, A_{j} \Bigr]\Bigr\}\dx  = - \intomega 2 \mubar\; \ubarn\; \delta u \dx
\end{split}
\end{equation*}
\begin{align*}
\intomega \Bigl[-\frac{1}{4\pi} \del \phibar \cdot \del \delta \phi +  (\ubarn^2 + b)\delta \phi \Bigr]\dx & = 0 \,,\\
\intomega \Bigl[\frac{1}{A_J\,B_J} \del \alpj\cdot\del\delta \betj + \frac{1}{A_J}\, \alpj \delta \betj + u^{2\alpha} \,\delta \betj \Bigr] \dx &= 0 \,,\\
\intomega \Bigl[\frac{1}{A_J\,B_J}\del \betj\cdot\del\delta \alpj + \frac{1}{A_J}\, \betj \delta \alpj  + u^{2\beta} \,\delta \alpj  \Bigr] \dx &= 0\,.
\end{align*}
Using equation ~\eqref{taylor_ker} and  the above Euler-Lagrange equations we get
\begin{equation}\label{errorkernel1}
\begin{split}
&E_h - E = \intomega \Bigl\{\frac{1}{2} |\del \delta u|^2 - 2 \mubar \ubarn\; \delta u  + \frac{1}{2} F''(\ubarn) (\delta u)^2
- \frac{1}{8\pi} |\del \delta \phi|^2 + (\delta u)^2\phibar + 2 \ubarn\; \delta u \;\delta \phi\\
&+ \sum_{J=1}^{m}\Bigl\{\frac{1}{A_J}\dalpjn \,\dbetjn + \frac{1}{A_J\,B_J}\del\dalpjn \cdot \del\dbetjn
+ 2\,\alpha(2\,\alpha - 1)\betj \,\ubarn^{2\alpha - 2} \frac{(\delta u)^2}{2}\\
& + 2\,\alpha \,\ubarn^{2\,\alpha - 1}\delta u \,\delta \betjn
+ 2\,\beta(2\,\beta - 1)\alpj \,\ubarn^{2\beta - 2} \frac{(\delta u)^2}{2} + 2\,\beta\, \ubarn^{2\,\beta - 1}\delta u \,\delta \alpjn \\
&+ 2\,(\alpha+ \beta)(2\,(\alpha+\beta) - 1)\,A_{J}\,\ubarn^{2(\alpha+\beta) - 2} \frac{(\delta u)^2}{2}\Bigr\}
\Bigr\}\dx.
\end{split}
\end{equation}
Using equations ~\eqref{multiplier} and ~\eqref{errorkernel1}, we arrive at the following bound in energy
\begin{equation}
\begin{split}
&|E_h - E| \leq \frac{1}{2} \left|\intomega |\del \delta u|^2 \dx\right| + \mubar \left|\intomega (\delta u)^2 \dx\right| + \frac{1}{2} \left|\intomega F''(\ubarn) (\delta u)^2 \dx \right|\\
&+ \frac{1}{8\pi} \left|\intomega |\del \delta \phi|^2 \dx \right|+ \left|\intomega (\delta u)^2\phibar \dx \right| + 2 \left|\intomega \ubarn\; \delta u \;\delta \phi \dx\right|
 + \sum_{J=1}^{m}\Bigl\{\left|\intomega\frac{1}{A_J}\dalpjn \,\dbetjn\dx\right | \\
 &+\left|\intomega \frac{1}{A_J\,B_J}\del\dalpjn \cdot \del\dbetjn\dx\right |
+ \left|\intomega 2\,\alpha(2\,\alpha - 1)\betj \,\ubarn^{2\alpha - 2} \frac{(\delta u)^2}{2}\dx\right|\\
 &+ \left|\intomega 2\,\alpha \,\ubarn^{2\,\alpha - 1}\delta u \,\delta \betjn\dx\right|
+ \left|\intomega 2\,\beta(2\,\beta - 1)\alpj \,\ubarn^{2\beta - 2} \frac{(\delta u)^2}{2}\dx\right| + \left|\intomega 2\,\beta\, \ubarn^{2\,\beta - 1}\delta u \,\delta \alpjn\dx\right| \\
&+ \left|\intomega 2\,(\alpha+ \beta)(2\,(\alpha+\beta) - 1)\,A_{J}\,\ubarn^{2(\alpha+\beta) - 2} \frac{(\delta u)^2}{2}\dx\right|\Bigr\}.
\end{split}
\end{equation}

We next find an optimal bound for the kernel terms involved in $|E_h - E|$.  As before, let $| \cdot |_{1,\Omega}$ represents the semi-norm in $H^{1}$ space,$\parallel\cdot \parallel_{1,\Omega}$  represents the $H^{1}$ norm, $\parallel\cdot \parallel_{0,\Omega}$  represents the standard $L^{2}$ norm, $\parallel\cdot \parallel_{0,p,\Omega}$  represents $L^{p}$ norm.
Firstly note that using H\"older inequality, we have
\begin{equation}
\begin{split}
\left|\intomega\frac{1}{A_J}\dalpjn \,\dbetjn\dx\right| &= \left|\frac{1}{A_J}\right| \left|\intomega \dalpjn \,\dbetjn\dx\right|\\
 &\leq C_1 \intomega \left| \dalpjn \,\dbetjn\right|\dx \leq C_1 \parallel \alpj - \alpj_h\parallel_{0,\Omega}\,\parallel \betj - \betj_h\parallel_{0,\Omega}
 \end{split}
\end{equation}
and
\begin{equation}
\begin{split}
\left|\intomega \frac{1}{A_J\,B_J}\del\dalpjn \cdot \del\dbetjn\dx\right| &= \left|\frac{1}{A_J\,B_J}\right|\left|\intomega \del\dalpjn \cdot \del\dbetjn\dx\right| \\
&\leq C_2 \intomega \left|\del\dalpjn \cdot \del\dbetjn\right| \dx \\
&\leq C_2 \parallel \del(\alpj - \alpj_h)\parallel_{0,\Omega}\,\parallel\del(\betj - \betj_h)\parallel_{0,\Omega}\\
&=C_2|\alpj - \alpj_h|_{1,\Omega}|\betj - \betj_h|_{1,\Omega}\,.
\end{split}
\end{equation}
Using H\"older inequality and the Sobolev inequality, we arrive at
\begin{equation}
\begin{split}
\left|\intomega 2\,\alpha(2\,\alpha - 1)\betj \,\ubarn^{2\alpha - 2} \frac{(\delta u)^2}{2}\dx\right| &\leq C _3\intomega \left|\betj \,\ubarn^{2\alpha - 2} (\delta u)^2\right|\dx \\
 &\leq C_3 \parallel \betj\,\ubarn^{2\alpha - 2} \parallel_{0,\Omega}\parallel(\ubarn - \ubarn_h)^2\parallel_{0,\Omega} \\
&= C_3 \parallel \betj \,\ubarn^{2\alpha - 2} \parallel_{0,\Omega} \parallel\ubarn - \ubarn_h\parallel_{0,4,\Omega}^{2}\\
&\leq \bar{C}_3  \parallel\ubarn - \ubarn_h\parallel^{2}_{1,\Omega}\,.
 \end{split}
\end{equation}
Further, note that
\begin{equation}
\begin{split}
\left|\intomega 2\,\alpha \,\ubarn^{2\,\alpha - 1}\delta u \,\delta \betjn\dx\right|&\leq C_4 \intomega \left|\ubarn^{2\,\alpha - 1}\delta u \,\delta \betjn \right|\dx \\
&\leq C_4\parallel\ubarn^{2\,\alpha - 1}\parallel_{0,6,\Omega}\parallel \ubarn-\ubarn_h\parallel_{0,\Omega}\parallel\betj - \betj_h\parallel_{0,3,\Omega}\\
&\leq \bar{C}_4 \parallel \ubarn-\ubarn_h\parallel_{0,\Omega}\parallel\betj - \betj_h \parallel_{1,\Omega}\,.
\end{split}
\end{equation}
  where we made use of generalized Holder inequality in the first step and Sobolev inequality in the next.
  Also one can show that
\begin{equation}
\begin{split}
\left|\intomega 2\,\beta(2\,\beta- 1)\alpj \,\ubarn^{2\beta - 2} \frac{(\delta u)^2}{2}\dx\right| &\leq C_5 \intomega \left|\alpj \,\ubarn^{2\beta- 2} (\delta u)^2\right|\dx \\
 &\leq C_5 \parallel \alpj \, \ubarn^{2\beta - 2}\parallel_{0,\Omega}\parallel(\ubarn - \ubarn_h)^2\parallel_{0,\Omega} \\
&= C_5 \parallel \alpj \, \ubarn^{2\beta - 2} \parallel_{0,\Omega} \parallel\ubarn - \ubarn_h\parallel_{0,4,\Omega}^{2}\\
&\leq \bar{C}_5  \parallel\ubarn - \ubarn_h\parallel^{2}_{1,\Omega}
 \end{split}
\end{equation}
and
\begin{equation}
\begin{split}
\left|\intomega 2\,\beta \,\ubarn^{2\,\beta - 1}\delta u \,\delta \alpjn\dx\right|&\leq C_6 \intomega \left|\ubarn^{2\,\beta - 1}\delta u \,\delta \alpjn \right|\dx \\
&\leq \bar{C}_6 \parallel \ubarn-\ubarn_h\parallel_{0,\Omega}\parallel\alpj - \alpj_h \parallel_{1,\Omega}\,.
\end{split}
\end{equation}
Using the bounds derived above, it follows that
\begin{equation}
\begin{split}
|E_h - E| &\leq C \Bigl[ \parallel\ubarn - \ubarn_h\parallel^{2}_{1,\Omega} + |\phibar - \phibar_h|^{2}_{1,\Omega} + \parallel\ubarn - \ubarn_h\parallel_{0,\Omega} \parallel\phibar - \phibar_h\parallel_{1,\Omega} \\
&+ \sum_{J=1}^{m}\Bigl\{\parallel \alpj - \alpj_h\parallel_{0,\Omega}\,\parallel \betj - \betj_h\parallel_{0,\Omega} + | \alpj - \alpj_h|_{1,\Omega}\,| \betj - \betj_h|_{1,\Omega}\\
&+\parallel \ubarn-\ubarn_h\parallel_{0,\Omega}\parallel\alpj - \alpj_h \parallel_{1,\Omega} + \parallel \ubarn-\ubarn_h\parallel_{0,\Omega}\parallel\betj - \betj_h \parallel_{1,\Omega}\Bigr\}\Bigr]\,.
\end{split}
\end{equation}

 Now in the neighbourhood of the solution ($\ubarn,\phibar,\mubar,\alpj,\betj$), we bound the above estimates with the interpolation error and then bound them with  finite-element mesh size in the similar lines of Section \ref{meshadapt}.
 \begin{subequations}
\begin{gather}
|\alpj - \alpj_h|_{1,\Omega} \leq \bar{C} |\alpj - \alpj_I|_{1,\Omega}\leq \tilde{C}\sum_{e} h_e^{k}|\alpj|_{k+1,\Omega_e}\,,\\
\parallel\alpj - \alpj_h\parallel_{0,\Omega} \leq \bar{C_0} \parallel\alpj- \alpj_I\parallel_{0,\Omega} \leq \tilde{C}_0 \sum_{e} h_e^{k+1}|\alpj|_{k+1,\Omega_e}\,,\\
|\betj - \betj_h|_{1,\Omega} \leq  \bar{\bar{C}} |\betj - \betj_I|_{1,\Omega}\leq \tilde{\tilde{C}} \sum_{e} h_{e}^{k}|\betj|_{k+1,\Omega_e}\,,\\
\parallel\betj - \betj_h\parallel_{0,\Omega} \leq \bar{\bar{C_0}}\parallel\betj- \betj_I\parallel_{0,\Omega} \leq \tilde{\tilde{C_0}}\sum_{e} h_{e}^{k+1}|\betj|_{k+1,\Omega_e}\,.
\end{gather}
\end{subequations}
Hence the error estimate in energy is given by
\begin{equation}\label{errorkernelfinal}
\begin{split}
&|E_h - E| \leq \mathcal{C}\sum_{e}\Bigl[ h_e^{2k}|\ubarn|_{k+1,\Omega_e}^{2}+ h_{e}^{2k}|\phibar|_{k+1,\Omega_e}^{2} + h_e^{2k+1}|\ubarn|_{k+1,\Omega_e}|\phibar|_{k+1,\Omega_e} \\
&+ \sum_{J=1}^{m}\Bigl\{h_{e}^{2k+2}|\alpj|_{k+1,\Omega_e}|\betj|_{k+1,\Omega_e} + h_{e}^{2k}|\alpj|_{k+1,\Omega_e}|\betj|_{k+1,\Omega_e}\\
 &+ h_{e}^{2k+1}|\ubarn|_{k+1,\Omega_e}|\betj|_{k+1,\Omega_e}+h_{e}^{2k+1}|\ubarn|_{k+1,\Omega_e}|\alpj|_{k+1,\Omega_e}\Bigr\}\Bigr]\,.
\end{split}
\end{equation}
The error estimate to $O(h^{2k+1})$ is therefore given by
\begin{equation}
|E_h - E| \leq \mathcal{C}\sum_{e}\Bigl[ h_{e}^{2k}|\ubarn|_{k+1,\Omega_e}^{2}+ h_{e}^{2k}|\phibar|_{k+1,\Omega_e}^{2}  + \sum_{J=1}^{m}\Bigl\{ h_{e}^{2k}|\alpj|_{k+1,\Omega_e}|\betj|_{k+1,\Omega_e}\Bigr\}\Bigr]
\end{equation}

%% file: solve.tex
If $X_h^{u}$ and $X_h^{\phi}$ represent the finite-dimensional subspaces with dimensions $n_u$ and $n_\phi$, the finite-element approximation for the field variables in problem~\eqref{infsup} is given by
\begin{align}
u_{h}(\bx) &= \sum_{k=1}^{n_u}N_{k}^{u}(\bx)u_k,\\
\phi_{h}(\bx) &= \sum_{k=1}^{n_\phi}N_{k}^{\phi}(\bx) \phi_k,
\end{align}
where  $N_{k}^{u}:1\leq k \leq n_u$ is the basis of $X_h^{u}$,  $N_{k}^{\phi}:1\leq k \leq n_\phi$ is the basis of $X_h^{\phi}$,  and $u_k$, $\phi_k$ are the nodal variables associated with the root-electron-density and the electrostatic potential. Let
\begin{equation}
L'(u_h,\phi_h,\mu_h) = L(u_h,\phi_h)+\mu_h(\intomega u_h^2 \dx-N)\,,
\end{equation}
where $L'(u_h,\phi_h,\mu_h)$ denotes the Lagrangian that imposes the constraint on the total number of electrons in the system via the Lagrange multiplier $\mu_h$.
The stationarity conditions associated with the discrete saddle-point problem in~\eqref{fem} are given by
\begin{subequations}
\begin{align}
\parder {L'} {u_i}(u_h,\phi_h,\mu_{h}) &= 0 \qquad i=1\ldots n_u\,,\\
\parder {L'} {\phi_i}(u_h,\phi_h) &= 0 \qquad i=1\ldots n_\phi\,, \\
\parder {L'} {\mu_{h}} (u_h) &= 0\,. \label{stn}
\end{align}
\end{subequations}
The above equations reduce to the following set of nonlinear equations:
\begin{subequations}\label{forceu}
\begin{align}
\begin{split}
&\sum_{k=1}^{n_u}\Bigl[\intomega \Bigl(\frac{\lambda}{2}\del\shp u k \cdot \del \shp u i + \shp u k \shp u i \phi_{h}(\bx)  \Bigr) \dx \Bigr]u_k + \intomega \frac{1}{2}F'(u_{h}) \shp u i  \dx \\
  &+ \mu_{h}\sum_{k=1}^{n_u}\; \left[ \intomega \shp u k \shp u i \dx \right]u_k  = 0 \qquad i=1\ldots n_u\,,
\end{split}
\end{align}
\begin{align}\label{forcephi}
&\sum_{k=1}^{n_\phi} \Bigl[\frac{1}{4\pi} \intomega \del \shp \phi i \del \shp \phi k \dx\Bigr]\phi_k - \intomega \left({u^2_{h}(\bx) }+ b(\bx,\bR)\right)\shp \phi i = 0 \qquad i=1\ldots n_\phi\,,
\end{align}
\begin{align}\label{forcemu}
&\sum_{s=1}^{n_u} \sum_{r=1}^{n_u} \Bigl[ \intomega \shp u r \shp u s \dx\Bigr]u_r u_s - N = 0\,,
\end{align}
\end{subequations}
where $F(u)$ is given by
\begin{equation*}
F(u) = C_F\;u^{10/3} +  \varepsilon_{xc}(u^2)u^2\,.
\end{equation*}
\subsection{Simultaneous Solution Procedure}
We now examine the solvability of  the above  system of  nonlinear equations in a simultaneous approach in which all the fields are solved in an iterative scheme using quasi-Newton methods. Firstly, the discrete set of nonlinear equations are linearized about the $n^{th}$ iteration, and the resulting incremental problem corresponds to a single step within the iterative scheme for the nonlinear problem. We argue that if each iterative update is uniquely solvable, then the scheme is well behaved. Denoting $(u_h^{(n)},\phi_h^{(n)},\mu_h^{(n)})$ to be the trial solution fields at the end of $n^{th}$ iteration, the linearizations of the stationarity conditions, neglecting second-order terms and beyond, are given by
\begin{subequations}\label{linstn}
\begin{align}
\parder {L'} {u_i}(u_h^{(n+1)}, \phi_h^{(n+1)}, \mu_{h}^{(n+1)}) =& \parder {L'} {u_i}\Big{|}_{(u_h^{(n)}, \phi_h^{(n)}, \mu_{h}^{(n)})} + \sum_{j=1}^{n_u} \parsecder {L'} {u_i} {u_j}\Big{|}_{(u_h^{(n)}, \phi_h^{(n)}, \mu_{h}^{(n)})} \delta u_j \notag\\
&+ \sum_{j=1}^{n_\phi} \parsecder {L'} {u_i} {\phi_j}\Big{|}_{u_h^{(n)}}  \delta \phi_j + \parsecder {L'} {u_i} {\mu_{h}}\Big{|}_{u_h^{(n)}} \delta \mu_{h} \quad i=1\ldots n_u\,,\\
\parder {L'} {\phi_i}(u_h^{(n+1)}, \phi_h^{(n+1)}) =& \parder {L'} {\phi_i}\Big{|}_{(u_h^{(n)}, \phi_h^{(n)})} + \sum_{j=1}^{n_u} \parsecder {L'} {\phi_i} {u_j}\Big{|}_{u_h^{(n)}} \delta u_j + \sum_{j=1}^{n_\phi} \parsecder {L'} {\phi_i} {\phi_j} \delta \phi_j \quad i=1\ldots n_\phi\,,\\
\parder {L'} {\mu_{h}}(u_h^{(n+1)}) =& \parder {L'} {\mu_{h}}\Big{|}_{u_h^{(n)}} + \sum_{j=1}^{n_u} \parsecder {L'} {\mu_{h}} {u_j}\Big{|}_{u_h^{(n)}} \delta u_j\,.
\end{align}
\end{subequations}
The corresponding matrix equation for the incremental solution at the $(n+1)^{th}$ iteration is given by
\begin{equation}\label{mateqn}
\bKbar \bxbar = \bffbar
\end{equation}
where
\begin{equation}
\bKbar = \left(
  \begin{array}{ccc}
    \bA^{(n)} & \bB^{(n)} & \bLam^{(n)} \\
    \bB^{(n)^{T}} & \bC & \mathbf{0} \\
    \bLam^{(n)^{T}} & \mathbf{0} & \mathbf{0} \\
  \end{array}
\right)\;\;\;\;\;\;\;\;\;
\bxbar = \left(
  \begin{array}{c}
    \delta \uhat \\
    \delta \phihat \\
    \delta  \mu \\
  \end{array}
\right)\;\;\;\;\;\;\;
\bffbar = \left(
  \begin{array}{c}
    \fhat^{(n)} \\
    \ghat^{(n)} \\
    h^{(n)} \\
  \end{array}
\right)
\end{equation}
and the matrix elements
\begin{equation}
A^{(n)}_{ij} = \parsecder {L'} {u_i} {u_j}\Big{|}_{(u_h^{(n)}, \phi_h^{(n)}, \mu_{h}^{(n)})}\,,\;\; B^{(n)}_{ij} = \parsecder {L'} {u_i} {\phi_j}\Big{|}_{u_h^{(n)}}\,,\;\; \Lambda^{(n)}_i = \parsecder {L'} {u_i} {\mu_{h}}\Big{|}_{u_h^{(n)}}\,,\;\;C_{ij} = \parsecder {L'} {\phi_i} {\phi_j}\,,
\end{equation}
are given by
\begin{subequations}
\begin{align}
\begin{split}
A_{ij} &= \frac{\lambda}{2}\intomega \del \shp u i \del \shp u j \dx + \frac{1}{2} \intomega F'' (u_{h}(\bx)) \shp u i \shp u j \dx\\
& + \intomega \phi_{h}(\bx)\shp u i \shp u j \dx + \mu_{h} \intomega \shp u i \shp u j \dx\,,
\end{split} \\
B_{ij} &= \intomega u_{h}(\bx) \shp u i \shp \phi j \dx \,, \\
\Lambda_{i} &= \intomega u_{h}(\bx) \shp u i \dx \,, \\
C_{ij} &= -\frac{1}{8\pi} \intomega \del \shp \phi i \del \shp \phi j \dx\,. \label{mat}
\end{align}
\end{subequations}
The components of force vectors
\begin{equation}
\hat{f}^{(n)}_{i} = -\parder {L'} {u_i}\Big{|}_{(u_h^{(n)}, \phi_h^{(n)}, \mu_{h}^{(n)})}\,, \;\;\;\;\hat{g}^{(n)}_{i} = -\parder {L'} {\phi_i}\Big{|}_{(u_h^{(n)}, \phi_h^{(n)})}\,,\;\;\;\;h^{(n)} = -\parder {L'} {\mu}\Big{|}_{u_h^{(n)}}\,,
\end{equation}
represent the forcing terms in the linearized stationarity conditions in~\eqref{linstn}.

The system matrices, $\bA^{(n)}$, $\bB^{(n)}$, $\bC$, $\bLam^{(n)}$, and vectors $\fhat^{(n)}$ , $\ghat^{(n)}$ and $h^{(n)}$ have the dimensions $n_u \times n_u$, $n_u \times n_\phi$, $n_\phi \times n_\phi$, $n_u \times 1$, and $n_u \times 1$, $n_\phi \times 1$, $1 \times 1$ respectively.
We observe that the  matrix $\bKbar$ is symmetric since $\bA^{(n)}$ and $\bC$ are symmetric, but it is indefinite and this relates to the saddle-point property of the two-field formulation~\eqref{fem}. Further, we note from equation~\eqref{mat} that $\bC$ is the discrete analog of the negative Laplace operator, and is hence symmetric negative definite and invertible. Also, $\bB^{(n)}$ is symmetric if function spaces $X_h^{u}$ and $X_h^{\phi}$ have the same basis.

The primary requirement for the robustness of the iterative solution is that $\bKbar$  must be non-singular at every iteration, which will guarantee a unique incremental solution of~\eqref{mateqn} and as the norm of the forcing vector $\|\bff\|$ tends to zero, the only solution admitted by $\bKbar \bxbar = \bffbar$ is the $\bzero$ vector. In order to study the invertibility conditions of equation~\eqref{mateqn}, we first solve the second partition in $\bKbar$ matrix for $\delta \phihat $ to obtain
\begin{equation}
\delta \phihat = \bC^{-1}\ghat^{(n)} - \bC^{-1}{\bB^{(n)}}^{T}\delta \uhat\,,
\end{equation}
and substitute in the first partition to get
\begin{subequations}
\begin{align}
(\bA^{(n)} - \bB^{(n)}\bC^{-1}{\bB^{(n)}}^{T})\;\delta \uhat + \bLam^{(n)} \; \delta  \mu &= \fhat^{(n)} - \bB^{(n)}\bC^{-1} \ghat\,,\\
{\bLam^{(n)}}^{T} \delta \uhat &= h^{(n)}\,.
\end{align}
\end{subequations}
Dropping the superscript $n$ in the expressions to follow for notational simplicity, the corresponding matrix equation is given by
\begin{equation}
\bK \boldsymbol{x} = \bff\,,
\end{equation}
where
\begin{equation}
\bK= \left(
        \begin{array}{cc}
          \bH & \bLam \\
          \bLam^{T} & 0 \\
        \end{array}
      \right)\,,\;\;\;\;\;\;\;\;
\boldsymbol{x} = \left(
        \begin{array}{c}
          \delta \uhat \\
          \delta \mu \\
        \end{array}
      \right)\,,\;\;\;\;\;\;
\bff = \left(
        \begin{array}{c}
          \bar{\fhat} \\
          h \,\\
        \end{array}
      \right),
\end{equation}
and
\begin{equation}
\bH = \bA - \bB\bC^{-1}\bB^{T}\,, \;\;\;\;\;\;\; \bar{\fhat} = \fhat - \bB\bC^{-1} \ghat\,.
\end{equation}
We now derive the solvability conditions of the matrix $\bK$ arising in our formulation. For a general discussion on the solvability and stability of mixed finite-element formulations, we refer to ~\cite{brezzi}.
We begin by defining the kernel spaces of $\bLam^{T}$,  $\bH$ and the column space of $\bLam$ in the following way
\begin{align*}
&Ker(\bLam^{T}) = \{\bv \in \mathbb{R}^{n_u} | \bLam^{T}\bv = 0\}, \\
&Ker(\bH) = \{\bv \in \mathbb{R}^{n_u} | \bH\bv = 0\} \label{null_H},\;\;\;\text{and}\\
&R = \{\bv \in \mathbb{R}^{n_u} | \bv = p \bLam\;\; \mbox{for some}\;\; p\in \mathbb{R}\}.
\end{align*}
By definition, the dimension of $R$ is 1, and by using rank-nullity theorem on $\bLam^{T}$ we note that the dimension of $Ker(\bLam^{T})$ is $n_u-1$. Further, we note that every vector $\bq\in R$ is orthogonal to every vector $\bv\in Ker(\bLam^{T})$. Since $dim(R)+dim(Ker(\bLam^{T}))=n_u$, we note that $R$ denotes the subspace of $\mathbb{R}^{n_u}$ containing all vectors that are orthogonal to $Ker(\bLam^{T})$.\\

\noindent We now state the following propositions:
\begin{prop}\label{prop1}
Let $\br \in Ker(\bLam^{T})$, then the following statements are equivalent:
\begin{description}
  \item[(i)]  $\bs^{T}\bH\:\br = 0 \;\;\;\forall \bs \in Ker(\bLam^{T})  \Rightarrow \br = 0$
  \item[(ii)] $\bK$ is invertible
  \end{description}
\end{prop}
\begin{proof}
We shall first show (i) $\Rightarrow$ (ii) and then show $\sim$ (i) $\Rightarrow$ $\sim$ (ii).  Let (i) be true and  let
\begin{equation}
 \boldsymbol{x} = \left(
         \begin{array}{c}
           \bl \\
           m \\
         \end{array}
       \right) \;\;\text{satisfy}\;\; \bK\boldsymbol{x} = \bzero\,,
 \end{equation}
 i.e., we seek the solutions $\bl$ and $m$ satisfying the following equations
 \begin{align}
 \bH\bl + \bLam m &= 0 \label{proof1}\,,\\
 \bLam^{T} \bl &= 0 \label{proof2}\,.
 \end{align}
 It is clear that any vector $\bl$ which solves the above system belongs to the space $Ker(\bLam^{T})$. Taking dot product with some vector $\bs \in Ker(\bLam^{T})$ on both sides of the equation~\eqref{proof1}, we get
 \begin{equation}\label{cond}
\bs^{T}\bH\:\bl + m\bs^{T}\bLam = 0
\Rightarrow \bs^{T}\bH\:\bl + m(\bLam^{T}\bs) ^{T} = \bzero
\Rightarrow \bs^{T}\bH\:\bl = 0\,.
\end{equation}
Since $\bs^{T}\bH\:\bl = 0$ is true for any $\bs \in Ker(\bLam^{T})$, (i) implies that $\bl=0$ is the only solution possible to~\eqref{cond}. Since $\bLam$ is a column vector, $m=0$ follows from equation~\eqref{proof1}. Thus, $\bK\boldsymbol{x} = 0$ admits $\boldsymbol{x} = \mathbf{0}$ as the only solution, and therefore it follows that $\bK$ is invertible.

Next let us suppose that (i) is not true, which implies that there exists a non-zero vector $\bl \in Ker(\bLam^{T})$ satisfying $\bs^{T}\bH\:\bl = 0 \;\;\;\forall \bs \in Ker(\bLam^{T}) $. This in turn implies that the vector $\bH\:\bl$ lies in $R$, the column space of $\bLam$. Thus, one can find $m$ such that $\bH\bl + \bLam m = 0$, and since $\bl \in Ker(\bLam^{T})$, it follows that $\bLam^{T} \bl = 0$. Hence, a non-zero vector $\boldsymbol{x} = (\bl\;\;\; m)$ satisfies $\bK \boldsymbol{x} = \bzero$ from which it follows that $\bK$ is not invertible. Thus, it follows that the statements (i) and (ii) are equivalent.
\end{proof}
\begin{remark}
If statement (i) in proposition~\ref{prop1} is true, then $Ker(\bH)\cap Ker(\bLam^{T}) = {\O}$. Further, if $\bH$ is either symmetric positive semi-definite or negative semi-definite, then the statement (i) is equivalent to $Ker(\bH)\cap Ker(\bLam^{T}) = {\O}$.
\end{remark}
\noindent We now turn to examine the restrictions imposed by the simultaneous approach on the choice of the interpolation spaces $X_h^u$ and $X_h^{\phi}$. \\

\begin{prop}\label{prop2}
If $\bK$ is invertible, rank($\bA$) $\geq$ $ n_u-(n_\phi+1)$
\end{prop}
\begin{proof}
We use Proposition \ref{prop1} to prove this.  Since $\bK$ is invertible, we have for any $\;\br \in Ker(\bLam^{T}) \;\; \text{and} \;\; \bs^{T}\bH\:\br = 0 \;\;\;\forall \bs \in Ker(\bLam^{T})$ admits $\br = 0$  as the only solution. From Proposition~\ref{prop1} $Ker(\bH)\cap Ker(\bLam^{T}) = \Phi$, and using the fact that the $dim(Ker(\bLam^{T}))$ is $n_u - 1$, it follows that $dim(Ker(\bH))\leq 1$. Hence, the $rank(\bH)\geq n_u-1$.  Letting $\bar{\bC} = -\bC$, we have
\begin{equation}
n_u - 1 \leq \;\;rank(\bH) \leq \;\;rank(\bA) + rank(\bB\bar{\bC}^{-1}\bB^{T})\label{rank}
\end{equation}
Using the fact that rank of $(\bB\bar{\bC}^{-1}\bB^{T})$  never exceeds the rank of $(\bC^{-1})$, it follows that $rank(\bB\bar{\bC}^{-1}\bB^{T}) \leq n_{\phi}$. Subsequently, using equation~\eqref{rank}, we conclude $rank(\bA) \geq n_u - (n_{\phi} + 1)$.
\end{proof}
\begin{remark}  We note that the nature of matrix $\bA$ is completely unknown at any given iteration, and may not be full-rank, since it depends on the initial guess of root-electron-density and electrostatic potential. Further, if the finite-element discretizations are chosen such that $n_u$ is exceedingly larger than $n_\phi$, it is possible that $n_u - (n_{\phi} + 1)$ exceeds $rank(\bA)$, and the system of equations will not be invertible. If the finite-element discretization is chosen such that $n_u \leq n_\phi$, the condition $rank(\bA) \geq n_u - (n_{\phi} + 1)$ is always satisfied. But, we note that this condition is only a necessary condition for $\bK$ to be invertible. The sufficiency condition for invertibility of $\bK$, given by (i) in Proposition~\ref{prop1}, is difficult to check for any general discretization and guess of the electronic fields.
\end{remark}
\noindent Next we analyze the staggered approach of solving the discrete orbital-free DFT problem.

\subsection{Staggered Solution Procedure}
We now examine the solvability of the discrete system of equations in~\eqref{forceu} using a staggered solution scheme. In this approach, we note that the stationarity condition $\parder {L'} {\phi_i}(u_h,\phi_h) = 0$  represents the discrete version of the Poisson equation, and solve the linear system of equations~\eqref{forcephi} to obtain the electrostatic potential consistently for a given root-electron-density. Thus, the root-electron-density and the Lagrange multiplier are solved in an iterative procedure using quasi-Newton methods, while the electrostatic potential is computed consistently for every root-electron-density using linear iterative solvers. The stationarity conditions corresponding to the staggered solution procedure are given by
\begin{subequations}
\begin{align}
\parder {L'} {u_i}(u_h,\phi_h(u_h),\mu_h) &= 0 \qquad i=1 \dots n_u\,, \\
\parder {L'} {\mu_{h}} (u_h) &= 0\,,
\end{align}
\end{subequations}
which can be expressed as the following matrix equation for the incremental problem
\begin{equation}\label{stagg}
\left(\begin{array}{cc}
          {\bA'}^{(n)} & \bLam^{(n)} \\
           \bLam^{(n)^{T}} & \mathbf{0} \\
        \end{array}
      \right)
\left(\begin{array}{c}
          \delta \uhat \\
          \delta \mu \\
        \end{array}
      \right)= \left(
        \begin{array}{c}
          \bar{\fhat} \\
          h \\
        \end{array}
      \right)\,.
\end{equation}
Consider the discrete form of the lagrangian in the orbital-free DFT problem written in the following way
\begin{equation}\label{min}
\begin{split}
L'(u_{h}) =& \inf_{u_{h} \in X_{h}^{u}} C_F\intomega ({u_{h}(\bx)})^{10/3}\dx + \frac{\lambda}{2}\intomega |\del u_{h}(\bx)|^{2} \dx \\
&+\intomega \varepsilon_{xc}([u_{h}(\bx)]^2)({u_{h}(\bx)})^2 \dx  + \Phi_h(u_{h}(\bx)) + \mu_h\left(\intomega ({u_{h}(\bx)})^{2} \dx - N\right)
\end{split}
\end{equation}
where
\begin{equation}\label{poisson}
\Phi_h(u_{h}(\bx)) = -\inf_{\phi_{h} \in X_{h}^{\phi}}\left[ \frac{1}{8\pi}\intomega |\del \phi_{h}(\bx)|^2 \dx - \intomega (({u_{h}(\bx)})^2 + b)\phi_{h}(\bx) \dx \right]\,.
\end{equation}
It follows from Proposition 3.3 of~\cite{ortner} that the discrete problem~\eqref{min} has a uniform minimizer in the neighborhood of a local minimizer of the continuous functional $L'(u)$ for sufficiently small $h$, the characteristic length of the mesh. Assuming that $L'(u_{h})$ is locally convex in a neighborhood of the minimizer, the positive definiteness of the matrix $\bA'$ follows, which in turn guarantees invertibility of the system matrix. We also note that the staggered solution scheme does not impose any restrictions on $n_u$ and $n_\phi$. Further, in our numerical studies, we find the staggered solution procedure is more robust in comparison to the  simultaneous approach, both in terms of  numerical convergence of the scheme and CPU time.

%% file: main.bbl
\newpage
\begin{thebibliography}{999}
\bibitem{kohn}
W. Kohn, L. J. Sham, Self-consistent equations including exchange and correlation effects, Phys. Rev. 140 (1965)  A1133--A1138.
\bibitem{Hohenberg}
P. Hohenberg, W. Kohn, Inhomogeneous electron gas, Phys. Rev. 136 (1964) B864--B871.
\bibitem{parr}
R. Parr, W. Yang, Density-functional theory of atoms and molecules, Oxford University Press, 2003.
\bibitem{Wang1992}
L. Wang, M.P. Teter, Kinetic energy functional of electron density, Phys. Rev. B 45 (1992) 13196--13220.
\bibitem{Enrico}
E. Smargiassi, P.A. Madden, Orbital-free kinetic-energy functionals for first-principle molecular dynamics, Phys. Rev. B  49 (1994) 5220--5226.
\bibitem{wang1}
Y. A. Wang, N. Govind, E. A. Carter, Orbital-free kinetic energy functionals for the nearly-free electron gas, Phys. Rev. B 58 (1998) 13465--13471.
\bibitem{wang2}
Y. A. Wang, N. Govind, E. A. Carter, Orbital-free kinetic energy density functionals with a density-dependent kernel, Phys. Rev. B 60 (1999) 16350--16358.
\bibitem{Huang}
C. Huang, E. A. Carter, Transferable local pseudopotentials for magnesium, aluminum and silicon, Phys. Chem. Chem. Phys. 10  (2008) 7109--7120.
\bibitem{VASP}
G. Kresse, J. Furthm\"uller, Efficient iterative schemes for ab initio total-energy calculations using a plane-wave basis set, Phys. Rev. B 54 (1996) 11169--11186.
\bibitem{CASTEP}
M. D. Segall, P. J. D. Lindan, M. J. Probert, C. J. Pickard, P. J. Hasnip, S. J. Clark, M. C. Payne, First-principles simulation: ideas, illustrations and the {CASTEP} code, J. Phys. Cond. Matt. 14 (2002) 2717--2744.
\bibitem{ABINIT}
X. Gonze, J. -M. Beuken, R. Caracas, F. Detraux, M. Fuchs, G. -M. Rignanese, L. Sindic, M. Verstraete, G. Zerah, F. Jollet, M. Torrent, A. Roy, M. Mikami, Ph. Ghosez, J. -Y. Raty, D. C. Allan, First-principles computation of material properties: the ABINIT software project, Comp. Mat. Sci. 25 (2002) 478--492.
\bibitem{PROFESS}
G. Ho, V. L. Ligneres, E. A. Carter, Introducing PROFESS: a new program for orbital-free density functional theory calculations, Comput. Phys. Commun. 179 (2008) 839--854.
\bibitem{Beck}
T.L. Beck, Real-space mesh techniques in density functional theory, Rev. Mod. Phys. 72 (2000) 1041--1080.
\bibitem{finnis}
M. Finnis, Interatomic forces in condensed matter, Oxford University Press, 2003.
\bibitem{white}
S. R. White, J. W. Wilkins, M. P. Teter, Finite element method for electronic structure, Phys. Rev. B 39 (1989) 5819--5830.
\bibitem{tsuchida1995}
E. Tsuchida, M. Tsukada,  Electronic-structure calculations based on the finite-element method, Phys. Rev. B 52 (1995) 5573--5578.
\bibitem{tsuchida1996}
E. Tsuchida, M. Tsukada, Adaptive finite-element method for electronic structure calculations,  Phys. Rev. B 54 (1996) 7602--7605.
\bibitem{tsuchida1998}
E. Tsuchida, M. Tsukada, Large-scale electronic-structure calculations based on the adaptive finite element method, J. Phys. Soc. Jpn. 67 (1998) 3844--3858.
\bibitem{pask1999}
J. E. Pask, B. M. Klein, C. Y. Fong, P. A. Sterne, Real-space local polynomial basis for solid-state electronic-structure calculations: a finite element approach, Phys. Rev. B  59 (1999) 12352--12358.
\bibitem{pask2001}
J. E. Pask, B. M. Klein, P. A. Sterne, C. Y. Fong, Finite element methods in electronic-structure theory, Comp. Phys. Comm. 135 (2001) 1--34.
\bibitem{carlos2006}
C. J. Garc$\acute{\text{i}}$a-Cervera, An efficient real-space method for orbital-free density functional theory, Comm. Comp. Phys. 2 (2006) 334-357.
\bibitem{ofdft}
 V. Gavini, J. Knap, K. Bhattacharya, M. Ortiz, Non-periodic finite-element formulation of orbital-free density functional theory, J. Mech. Phys. Sol. 55 (2007) 669--696.
 \bibitem{qcofdft}
V. Gavini,  K. Bhattacharya, M. Ortiz, Quasi-continuum orbital-free denisty functional theory:a route to multi-million atom non-periodic DFT calculation, J. Mech. Phys. Sol. 55 (2007) 697--718.
\bibitem{Zhou2008}
D. Zhang, L. Shen, A. Zhou, X. Gong, Finite element method for solving Kohn-Sham equations based on self-adaptive tetrahedral mesh, Phys. Lett. A 372 (2008) 5071--5076.
\bibitem{suryanarayana2010non}
P. Suryanarayana, V. Gavini, T. Blesgen, K. Bhattacharya, M. Ortiz, Non-periodic finite-element formulation of Kohn-Sham density functional theory, J. Mech. Phys. Sol. 58 (2010) 256--280.
\bibitem{Lin}
L. Lin, J. Lu, W. E., Adaptive local basis set for Kohn-Sham density functional theory in a discontinuous Galerkin framework I: Total energy calculation, arXiv:1102.2520.
\bibitem{bylaska}
E. J. Bylaska, M. Host, J. H. Weare., Adaptive Finite Element Method for Solving the Exact Kohn-Sham Equation of Density Functional Theory, J. Chem. Theory. Comput. 5 (2009) 937-948.
\bibitem{lehtovaara}
L. Lehtovaara, V. Havu, M. Puska., All-electron density functional theory and time-dependent density functional theory with high-order finite elements, J. Chem. Physics. 131 (2009) 054103.
\bibitem{Hermannson}
B. Hermannson, D. Yevick, Finite-element approach to band-structure analysis, Phys. Rev. B 33 (1986) 7241-7242.
\bibitem{batcho}
P. F. Batcho., Computational method for general multicenter electronic structure calculations, Phys. Rev E. 61(2000) 7169-7183.
\bibitem{ortner}
 B. Langwallner, C.Ortner, E.Sulli, Existence and convergence results for the Galerkin approximation of an electronic density functional, Math. Models Meth. Appl. Sci. 20 (2010) 2237--2265.
\bibitem{zhou}
H. Chen, X. Gong, A. Zhou, Numerical approximations of a nonlinear eigenvalue problem and applications to a density functional model, Math. Meth. Appl. Sci. 33 (2010) 1723--1742.
\bibitem{Cances}
E. Canc$\grave{e}$s, R. Chakir, Y. Maday, Numerical analysis of nonlinear eigenvalue problems, J. Sci. Comput. 45 (2010) 90--117.
\bibitem{pask2005}
J. E. Pask,  P. A. Sterne, Finite element methods in ab initio electronic structure calculations, Mod. Sim. Mat. Sci. Eng. 13 (2005) R71--R96.
\bibitem{bala}
B. Radhakrishnan, V. Gavini., Effect of cell size on the energetics of vacancies in aluminum studied via orbital-free density functional theory, Phys. Rev. B 82 (2010) 094117-094121.
\bibitem{radio}
R. A. Radovitzky, Error estimation and adaptive meshing in strongly nonlinear dynamic problems, Ph.D. Thesis, California Institute of Technology, (1998).
\bibitem{patera1984spectral}
A.T. Patera, A spectral element method for fluid dynamics: laminar flow in a  channel expansion, J. Comp. Phys. 54 (1984) 468--488.
\bibitem{alder}
 D. M. Ceperley, B. J. Alder, Ground state of the electron gas by a stochastic method, Phys. Rev. 45 (1980) 566--569.
\bibitem{perdew}
 J. P. Perdew, A. Zunger, Self-interaction correction to density-functional approximation for many-electron systems, Phys. Rev. B 23 (1981) 5048--5079.
\bibitem{lieb}
E. H. Lieb, Thomas-fermi and related theories of atoms and molecules, Rev. Mod. Phy. 53 (1981) 603--641.
\bibitem{kaxiras}
N. Choly, E. Kaxiras, Kinetic energy density functionals for non-periodic systems, Sol. St. Comm. 121 (2002) 281--286.
\bibitem{brezzi}
F. Brezzi, K. J. Bathe, A discourse on the stability conditions for mixed finite element formulations, Comp. Meth. App. Mech. Eng. 82 (1990) 27--57.
\bibitem{Ciarlet}
P.G. Ciarlet, The finite element method for elliptic problems, SIAM, Philadelphia, 2002.
\bibitem{boyd2001chebyshev}
 J. P. Boyd, Chebyshev and Fourier spectral methods, Dover Publications, 2001.
 \bibitem{shewchuk1994}
J. R. Shewchuk, An introduction to the conjugate gradient method without the agonizing pain, 1994.
\bibitem{collier2006}
A. M. Collier, A. C. Hindmarsh, R. Serban, C. S. Woodward, User documentation for KINSOL v2. 4. 0, 2006.
\bibitem{dennis1996}
J. E. Dennis, R. B. Schnabel, Numerical methods for unconstrained optimization and nonlinear equations, SIAM, 1996.
\bibitem{saad1986gmres}
Y. Saad, M. H. Schultz, GMRES: A generalized minimal residual algorithm for solving nonsymmetric linear systems, SIAM J. Sci. Stat. Comput. 7 (1986) 856--869.
\bibitem{van1992bi}
H. A. Vander Vorst,  Bi-CGSTAB: A fast and smoothly converging variant of Bi-CG for the solution of nonsymmetric linear systems, SIAM J. Sci. Stat. Comput. 13 (1992) 631--644.
\bibitem{freund1993}
R.W. Freund, A transpose-free quasi-minimal residual algorithm for non-Hermitian linear systems, SIAM J. Sci. Comput. 14 (1993) 470--482.
\bibitem{chowhypre}
E. Chow, A. Cleary, R. Falgout, HYPRE User’s manual, v1.6.0. Technical Report UCRL-MA-137155, Lawrence Livermore National Laboratory, Livermore, CA, 1998.
\bibitem{goodwin1990}
L. Goodwin, R. J. Needs, V. Heine, A pseudopotential total energy study of impurity-promoted intergranular embrittlement, J. Phys: Cond. Mat. 2 (1990) 351--365.
\end{thebibliography}
